\documentclass{article}

\usepackage[utf8]{inputenc}
\usepackage{graphicx,authblk}
\usepackage{caption}
\usepackage{subcaption}
\usepackage{xcolor}
\usepackage{mathtools}
\usepackage{amsmath}
\usepackage{amssymb}
\usepackage{amsthm}
\usepackage{xfrac}
\usepackage{bm}
\usepackage{enumitem}
\usepackage{hyperref}
\usepackage[lined,boxed,commentsnumbered,ruled,vlined]{algorithm2e}
\usepackage{ifthen}
\usepackage{wasysym}
\usepackage{bbold}
\usepackage[round]{natbib}

\usepackage{todonotes}

\makeatletter
\newcommand{\bigcomp}{%
  \DOTSB
  \mathop{\vphantom{\sum}\mathpalette\bigcomp@\relax}%
  \slimits@
}
\newcommand{\bigcomp@}[2]{%
  \begingroup\m@th
  \sbox\z@{$#1\sum$}%
  \setlength{\unitlength}{0.9\dimexpr\ht\z@+\dp\z@}%
  \vcenter{\hbox{%
    \begin{picture}(1,1)
    \bigcomp@linethickness{#1}
    \put(0.5,0.5){\circle{1}}
    \end{picture}%
  }}%
  \endgroup
}
\newcommand{\bigcomp@linethickness}[1]{%
  \linethickness{%
      \ifx#1\displaystyle 2\fontdimen8\textfont\else
      \ifx#1\textstyle 1.65\fontdimen8\textfont\else
      \ifx#1\scriptstyle 1.65\fontdimen8\scriptfont\else
      1.65\fontdimen8\scriptscriptfont\fi\fi\fi 3
  }%
}

 \DeclareMathOperator*{\argmax}{arg\,max}

\newtheorem*{rep@theorem}{\rep@title}
\newcommand{\newreptheorem}[2]{%
\newenvironment{rep#1}[1]{%
 \def\rep@title{#2 \ref{##1}}%
 \begin{rep@theorem}}%
 {\end{rep@theorem}}}
\makeatother
\newreptheorem{theorem}{Theorem}

\newtheorem*{rep@proposition}{\rep@title}
\newcommand{\newrepproposition}[2]{%
\newenvironment{rep#1}[1]{%
 \def\rep@title{#2 \ref{##1}}%
 \begin{rep@proposition}}%
 {\end{rep@proposition}}}
\makeatother
\newreptheorem{proposition}{Proposition}

\newtheorem*{rep@corollary}{\rep@title}
\newcommand{\newrepcorollary}[2]{%
\newenvironment{rep#1}[1]{%
 \def\rep@title{#2 \ref{##1}}%
 \begin{rep@corollary}}%
 {\end{rep@corollary}}}
\makeatother
\newreptheorem{corollary}{corollary}

\newtheorem*{rep@algorithm}{\rep@title}
\newcommand{\newrepalgorithm}[2]{%
\newenvironment{rep#1}[1]{%
 \def\rep@title{#2 \ref{##1}}%
 \begin{rep@algorithm}}%
 {\end{rep@algorithm}}}
\makeatother
\newreptheorem{algorithm}{algorithm}

\newtheorem*{rep@definition}{\rep@title}
\newcommand{\newrepdefinition}[2]{%
\newenvironment{rep#1}[1]{%
 \def\rep@title{#2 \ref{##1}}%
 \begin{rep@definition}}%
 {\end{rep@definition}}}
\makeatother
\newreptheorem{definition}{Definition}

\newtheorem*{rep@lemma}{\rep@title}
\newcommand{\newreplemma}[2]{%
\newenvironment{rep#1}[1]{%
 \def\rep@title{#2 \ref{##1}}%
 \begin{rep@lemma}}%
 {\end{rep@lemma}}}
\makeatother
\newreptheorem{lemma}{Lemma}

\newtheorem{theorem}{Theorem}
\newtheorem{lemma}{Lemma}
\newtheorem{corollary}{Corollary}

\newtheorem{definition}{Definition}

\usepackage{thmtools, thm-restate}

 \SetKwComment{Comment}{/* }{ */}

\usepackage{versions}

\newcommand{\userdic}{\texttt{TBS}}

\newcommand{\bsearch}{\texttt{SEARCH}}
\newcommand{\bisolate}{\texttt{ISOLATE}}
\newcommand{\bsplit}{\texttt{SPLIT}}
\newcommand{\bclear}{\texttt{UPDATE}}

\newcommand{\idxleft}{l}
\newcommand{\idxright}{r}
\newcommand{\request}{\texttt{Query}}
\newcommand{\rating}[2]{q({#1},{#2})}
\newcommand{\dvg}[2]{\E\left[\frac{f_{#1}({#2})^2}{f_{#2}({#2})^2} \right]}

\newcommand{\binplus}{\bin^+}
\newcommand{\binminus}{\bin^-}
\newcommand{\tsearch}{\sampcomp^{search}}
\newcommand{\tiso}{\sampcomp^{iso}}
\newcommand{\tsplit}{\sampcomp^{split}}

\newcommand{\nbi}{n}

\newcommand{\nbu}{m}
\newcommand{\nbur}{M}

\newcommand{\thresh}{Y}
\newcommand{\threshord}[1]{\thresh_{(#1)}}
\newcommand{\score}{X}
\newcommand{\scoreord}[1]{\score_{(#1)}}

\newcommand{\idxitem}{i}
\newcommand{\idxitemb}{j}
\newcommand{\idxuser}{u}

\newcommand{\idxbin}{k}
\newcommand{\idxbinb}{l}

\newcommand{\sampcomp}{Q}

\newcommand{\R}{\mathbb{R}}

\newcommand{\E}{\mathbb{E}}
\newcommand{\N}{\mathbb{N}}
\newcommand{\V}{\mathbb{V}}
\newcommand{\prob}{\mathbb{P}}

\newcommand{\rv}{Z}

\newcommand{\permutset}{\mathcal{S}}

\newcommand{\msf}{F}
\newcommand{\msffunc}{\texttt{MSF}}
\newcommand{\sffunc}{\texttt{SF}}

\newcommand{\isodd}{\text{ is odd}}

\newcommand{\bound}{g_\nbu}

\newcommand{\bin}{\mathcal{B}}
\newcommand{\binseq}{\mathfrak{B}}
\newcommand{\event}{\mathcal{E}}
\newcommand{\betaf}{\mathbf{B}}

\usepackage[preprint]{neurips_2025}

\usepackage[utf8]{inputenc} 
\usepackage[T1]{fontenc}    
\usepackage{hyperref}       
\usepackage{url}            
\usepackage{booktabs}       
\usepackage{amsfonts}       
\usepackage{nicefrac}       
\usepackage{microtype}      
\usepackage{xcolor}         

\title{Ranking Items from Discrete Ratings:\\
The Cost of Unknown User Thresholds}

\author{
Oscar Villemaud \\
\texttt{oscar.villemaud@epfl.ch} \\
\and
Suryanarayana Sankagiri \\
\texttt{suryanarayana.sankagiri@epfl.ch} \\
\and
 Matthias Grossglauser\\
\texttt{matthias.grossglauser@epfl.ch} \\
}

\begin{document}

\maketitle

\begin{abstract}
Ranking items is a central task in many information retrieval and recommender systems.
User input for the ranking task often comes in the form of ratings on a coarse discrete scale.
We ask whether it is possible to recover a fine-grained item ranking from such coarse-grained ratings.
We model items as having scores and users as having thresholds; a user rates an item positively if the item’s score exceeds the user's threshold.
Although all users agree on the total item order, estimating that order is challenging when both the scores and the thresholds are latent.
Under our model, any ranking method naturally partitions the $n$ items into bins; the bins are ordered, but the items inside each bin are still unordered.
Users arrive sequentially, and every new user can be queried to refine the current ranking.
We prove that achieving a near-perfect ranking, measured by Spearman distance, requires $\Theta(n^2)$ users (and therefore $\Omega(n^2)$ queries).
This is significantly worse than the $O(n\log n)$ queries needed to rank from comparisons; the gap reflects the additional queries needed to identify the users who have the appropriate thresholds.
Our bound also quantifies the impact of a mismatch between score and threshold distributions via a quadratic divergence factor.
To show the tightness of our results, we provide a ranking algorithm whose query complexity matches our bound up to a logarithmic factor.
Our work reveals a tension in online ranking: diversity in thresholds is necessary to merge coarse ratings from many users into a fine-grained ranking, but this diversity has a cost if the thresholds are a priori unknown.
\end{abstract}

\section{Introduction}

Ranking items according to human preferences is a central task underlying many platforms: search engines order links, recommender systems curate content, and peer review selects papers for publication. In many such applications, the feedback signal is a discrete score for each item: a binary action like a click or a like, or ratings on a small integer scale. In this paper, we ask the question: how hard is it to obtain fine-grained rankings from coarse ratings? 

Surprisingly, this fundamental question has received scant attention in the theoretical machine learning literature. Prior work largely models a discrete rating as a noisy signal of the item's score (\textit{i.e.}, utility). In contrast, we view discretization as the result of a thresholding process. That is, we assume that each user has a \textit{latent discretization threshold}, and they rate two items differently only when this threshold lies between the items' scores. Thus, a platform can order two items only by querying a user who has an appropriately placed threshold. In most circumstances, this discretization threshold is apriori unknown to the platform, which makes it difficult to order and rank items. In contrast, asking users to compare items totally eliminates the effect of discretization: any user can tell the order between two items. Indeed, in many practical learning-to-rank task, it is commonly observed that asking users to compare items is better than soliciting ratings.

There are several interesting phenomena that emerge in the ranking task as a consequence of the discretization.
Firstly, in order to be able to rank all items, we need a diverse population of users with different discretization thresholds.
In particular, the population of thresholds, as it were, should be dense in the same region where the items are closely spaced.
Secondly, there is a natural phenomenon of diminishing returns: the chance that a new user is able to shed new light on the order of the two items goes down over time.
Finally, asking users to rate items at random is inefficient.
Instead, it is better to ask users to rate items that they are most likely to discriminate; for this, one must estimate the user's latent discretization threshold.

In this work, we quantify the aforementioned phenomena by developing a simple, probabilistic model of a platform that aims to rank items from user-provided ratings. We assume there are $n$ items, with each item $i$ having a score $X_i$ in the interval $[0,1]$. Further, we assume these scores are distributed i.i.d. according to a distribution $f_X$. An item with a higher score is considered better than an item with a lower score. We assume users arrive sequentially to the system, and the platform adaptively chooses which items to query to the user. Each user $u$ has a latent discretization threshold $Y_u \in [0,1]$. When asked to rate an item $i$, user $u$ rates it as $1$ if and only if the item's score exceeds the user's threshold ($X_i > Y_u$). We assume each $Y_u$ is distributed i.i.d. according to $f_Y$. As the user feedback is noiseless, we assume that each item is queried at most once per user. For simplicity, we assume users do not return to the system once they have finished answering the platform's rating queries. 

The aforementioned model implicitly assumes that all users agree on the order of the items, even though they may have different \textit{rating styles}; users with a low threshold are more lenient, while users with a larger threshold are more stringent. The model makes it clear why multiple users are necessary; the first user can split all items into two bins: those rated one and those rated zero, and every successive user further splits existing bins. Thus, at any given stage, the platform maintains a \textit{partial ranking} over the items, successively refining it with every new user. 

A basic question is, how many users must the platform see in order to recover the underlying ranking? Clearly, this number is a random variable, dependent on the thresholds and scores; in particular, it is a stopping time. Somewhat surprisingly, we find that in expectation, the number of users needed for a perfect ranking is infinity, \textit{irrespective of the number of items} (see Lemma \ref{lem:infinite}). This is because the corresponding distribution is heavy-tailed. Consequently, we look to quantify the accuracy of the partial ranking, given a fixed number of users. In this work, we measure the accuracy in terms of the Spearman footrule distance, a popular metric for comparing two rankings.

Our main results, Theorems \ref{thm:msf_high} and \ref{thm:msf_low}, identify the complexity of ranking in two different regimes. Theorem \ref{thm:msf_high} shows that if the number of users $m$ scales linearly with the number of items $n$, then, on average, the partial ranking learned from the platform can differ from the ground-truth ranking by a footrule distance that is $O(n)$. Loosely, this means that each item's rank can be identified to within a constant distance. Theorem \ref{thm:msf_low} extends this analysis and shows that if the number of users $m$ scales super-linearly with $n$, then the average footrule distance scales as $O(n^2/m)$. In particular, if $m = \Omega(n^2)$, then the platform can rank all but a few items correctly. In other words, to obtain a near-perfect ranking of $n$ items, one needs $\Omega(n^2)$ users. Our analysis also shows that misalignment between item scores and threshold distribution increases complexity; in both theorems, the dominant term contains a quadratic divergence factor $\mathbb{E}[(f_X(Y)/f_Y(Y))^2]$, which is the least when $f_X$ and $f_Y$ are identical, and increases as the two distributions drift apart.

These results should be interpreted as lower bounds on the complexity of ranking from discrete ratings. In particular, Theorem \ref{thm:msf_low} suggests that the number of users needed to obtain a near-perfect order of $n$ items scales as $\Omega(n^2)$. The number of users is clearly a lower bound for the sample (\textit{i.e.}, query) complexity; indeed, each user must be queried at least once in order for their presence to count. This sample complexity result should be contrasted with the complexity of ranking from pairwise comparisons. It is well-known that sorting algorithms can obtain a complete ranking with $O(n\log n)$ (adaptively-chosen) pairwise comparisons. Thus, our theoretical analysis shows that ranking items by soliciting user ratings is much harder than ranking them by asking users to compare items. Moreover, as our model assumes universal agreement on item order and noiseless ratings, this extra factor of $n$ in the sample complexity can be attributed purely as the cost of discretization. 

In addition to the theoretical analysis outlined above, we provide insights into the tightness of our lower bounds by designing and analyzing an algorithm, called threshold binary search (TBS). The algorithm uses a binary search--style strategy to steer each user’s queries toward the most informative regions. We give a proof sketch of the algorithm's complexity, and find that it matches the lower bounds up to log factors. These findings are also corroborated through simulation results.

In conclusion, our work provides the first clean theoretical explanation for a widely observed empirical fact: comparisons are more effective than ratings for ranking. The hardness we uncover is fundamental and persists under alternative metrics. These results carry practical implications. If fine-grained rankings are needed, it is better to solicit comparisons rather than ratings. If ratings must be used, systems should ensure diversity across users’ thresholds and design prompts so that users are most discriminative in regions where items are concentrated. We discuss extensions to $k$-ary ratings and noisy ratings in the discussion section of the paper.

\paragraph{Structure of the paper}
Section \ref{sec:model} formalizes our model and presents the problem.
Section \ref{sec:msf} proves a relation between the number of users and the expected precision of the ranking.
In Section \ref{sec:algorithm}, we study the tightness of our results by providing an algorithm that efficiently orders items from binary ratings.
Section \ref{sec:experiments} shows experimental results
that support the conclusions of the previous two sections.
Finally, we discuss related work and conclude in Section \ref{sec:conclusion}.

\section{Model and Problem}
\label{sec:model}

\subsection{Model and Definitions}

 We consider a model with a finite number $\nbi$ of items.
Each item $\idxitem \in [\nbi]$ has an unknown score $\score_\idxitem$ that represents its utility. 
The item scores $\score_\idxitem$ are \emph{iid} random variables on $[0,1]$, of density $f_X$.
There is a potentially infinite number of users who give feedback on the items by rating them.
Although the item scores are common to all users, the rating of different users for the same item may differ because different users have different \emph{rating styles}, or biases.
This bias manifests itself in the form of a unique discretization threshold for ratings. 
Thus, some users rate most items as $1$ while others rate most items as $0$.

The rating style of a user $\idxuser \in \N$ is represented by an individual threshold $\thresh_\idxuser \in [0,1]$ that controls their rating in the following way: 
$$ \forall \idxuser \in \N, \forall \idxitem \in [\nbi], \rating{\idxuser}{\idxitem} \triangleq \mathbb{1} (\score_\idxitem > \thresh_\idxuser )$$

where $\rating{\idxuser}{\idxitem}$ is the binary rating given to item $\idxitem$ by user $\idxuser$. 
Under this model, the item score is the probability that a random user rates the item as $1$.
The thresholds $\thresh_\idxuser$ are \emph{iid} uniform random variables in $[0,1]$, of density $f_Y$.
We assume that $f_X$ and $f_Y$ are non-zero, upper bounded and $c$-Lipschitz.

\subsection{The Cost Of Exact Ranking}


We assume that users arrive in a sequence and are asked to rate some selected items according to an algorithm. 
Given a set of items, our goal is to study how many users and queries we need to order the items. 
Consider a fixed set of users $E$. 
Under our model, it is possible to relatively order two items $\idxitem$ and $\idxitemb$ if and only if there exists a user $\idxuser \in E$ such that $\score_\idxitem < \thresh_\idxuser < \score_\idxitemb$ (or $\score_\idxitem > \thresh_\idxuser > \score_\idxitemb$). 
This means that a necessary condition to fully order the items is to have a user threshold between all pairs of consecutive item scores.

However, the items scores are \emph{iid} on $[0,1]$. 
In this case, the number of users needed to order all items is infinite in expectation,
as stated by the following lemma:

\begin{lemma}
\label{lem:infinite}
Let $\score_1, \ldots, \score_\nbi$ be iid item scores of density $f_X$ on $[0,1]$.
    Let  $\nbur$ be the random number of users needed to obtain a total order.  Then we have:
    $$ \E[\nbur] = \infty $$
\end{lemma}

We formally prove this result in Appendix \ref{app:sec:perfect}, but wish to provide some intuition here.
For the sake of simplicity, we take the case  $f_X=f_Y=1$, \emph{i.e.} item scores and user thresholds are uniform on $[0,1]$.
Consider an interval between two consecutive items.
Conditional on this interval, the expected number of user thresholds we need to try until one falls in this interval and separates the two items is inversely proportional to its length.
The problem arises because the distribution of the interval length is a Beta of parameter one, \emph{i.e.}, does not vanish at zero.
In other words, we are too likely to get at least one very short interval, which dominates the total cost.
In fact, if the items were instead equally spaced in the $[0,1]$ interval, the number of thresholds would correspond to the classical coupon collector problem.
In that case, the expected number of users would be $\E[\nbur] = \nbi \log(\nbi)$.

In summary, Lemma \ref{lem:infinite} suggests that a more interesting regime is where we assume a finite number of users, and ask how well we can approximate the ground-truth ranking.
In order to formalize the notion of partial order for our problem, we introduce a new structure (the bin sequence) and an associated error metric (the MSF).

As stated before, we need a threshold between each pair of consecutive item scores in order to fully order the items, and this requires an infinite expected number of users.
We now assume that we have a finite number of users, and we observe what we can say about the order of the items.
Consider the first user and assume they rate all items.
This gives us the information about which items have their score below $\thresh_1$ and which have their score above.
This divides the items in two groups such that we know the relative ordering of a pair of items $(\idxitem, \idxitemb)$ \emph{if and only if} $\idxitem$ and $\idxitemb$ are not in the same group.
If $\idxitem$ and $\idxitemb$ are in the same group, we have no information on the order of their scores.
This is the maximum amount of information the first user can provide.
With more users, we are able to divide the items in more groups, but items within the same group are still indistinguishable.
In the rest of the paper, we refer to these groups of items as \emph{bins}, and to the ordered set of bins as a \emph{bin sequence}.

\begin{definition}
\label{def:binseq}
    \textbf{Bin Sequence and Bins.} We call \textbf{bin sequence} an ordered partition $\binseq = (\bin_1, ..., \bin_{|\binseq|})$ of $[\nbi]$ that respects the order of the item scores.  \\
    Formally :   
    $$\forall \idxbin < \idxbin' \in [|\binseq|], \forall \idxitem \in \bin_\idxbin, \idxitem' \in \bin_{\idxbin'}, \score_\idxitem < \score_{\idxitem'}$$ 
   $\bin_1, \ldots, \bin_{|\binseq|}$ are called \textbf{bins} of items.
   They are sets of items whose scores belong to a certain interval.\\
     We note $\permutset_\binseq$ the partial order induced by this bin sequence, \emph{i.e.} $\permutset_\binseq$ is the set of orderings compatible with $\binseq$.
\end{definition}

Figure \ref{fig:model} describes our model.

 \begin{figure}[ht!]
  \vspace{-5pt}
 \centering
   \includegraphics[width=\linewidth]{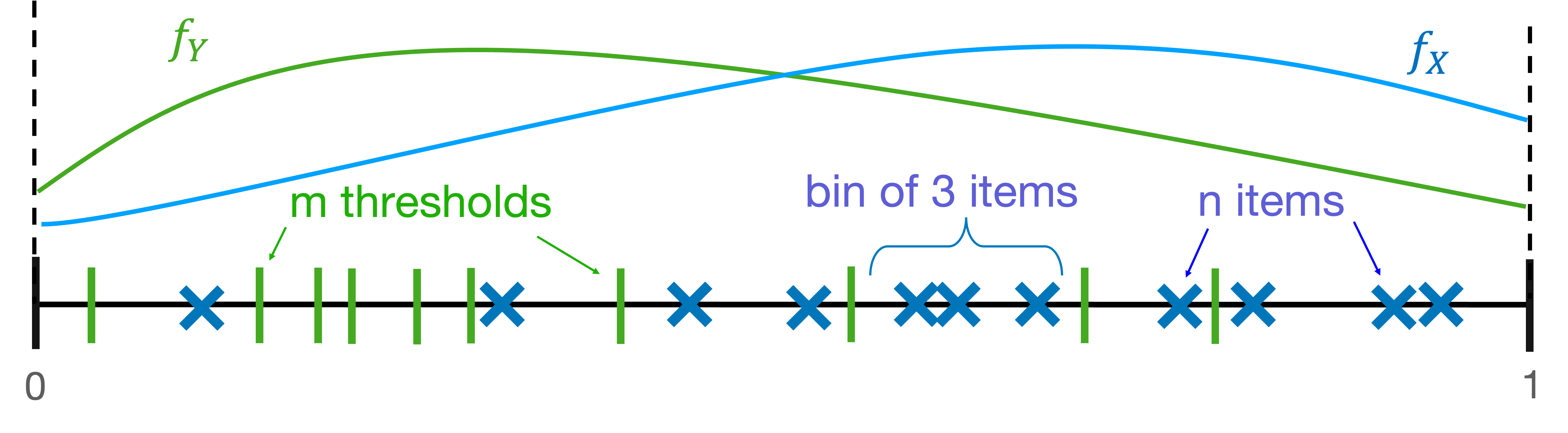}
  \vspace{-10pt}
  \caption{The item scores and user thresholds are sampled \emph{iid}, respectively with densities $f_X$ and $f_Y$.}
 \label{fig:model}
 \vspace{-10pt}
 \end{figure}

In the following, when we talk about partial order, we always refer to a partial order induced by a bin sequence.
In order to define the precision of a partial order, we rely on the \emph{Spearman Footrule} (SF) distance between two permutations to define the \emph{Maximum Spearman Footrule} (MSF).  
For any finite subset $\bin$ of $\N$, let $\permutset_\bin$ denote the set of possible orders of $\bin$.
For any order $\sigma \in \permutset_\bin$ and for any $i \in \bin$, $\sigma(i)$ denotes the rank of item $i$ among the elements of $\bin$.
 Then the SF distance between two orders on $\bin$ is defined as:

\begin{definition}
\label{def:sf}
\textbf{Spearman Footrule.} 
    $\forall \sigma, \sigma^* \in \permutset_\bin$,  $\sffunc(\sigma, \sigma^*) \triangleq \sum_{i \in \bin} |\sigma(i) - \sigma^*(i)|$
\end{definition}

Given a partial order, the MSF measures the worst (\emph{i.e.} maximum) Spearman Footrule between two total orders compatible with the partial order:

\begin{definition} 
\label{def:msf}
\textbf{Maximum Spearman Footrule (MSF)}.\\
Let $\bin \subseteq \N$ be  a finite set. 
Let $\permutset$ be a subset of $\permutset_\bin$. 
Then,  
    $$ \msffunc(\permutset) 
    \triangleq \max_{\sigma, \sigma' \in \permutset} \sffunc(\sigma, \sigma')$$
\end{definition}

Under this definition, the MSF is the ``diameter'' of the set of possible orderings, with respect to the Spearman Footrule distance. 
In particular, if the true ordering $\sigma^*$ belongs to $\permutset$, then we have\\
$\forall \sigma \in \permutset, \quad \sffunc(\sigma, \sigma^*)  \leq \msffunc(\permutset)$


In the rest of the paper, we study how many users and items are needed to obtain a constant expected MSF.
A constant MSF means that most items are correctly  placed, and only a constant number of them has uncertainty on their rank.
Allowing for this vanishing fraction of items not to be fully ordered is what allows us to overcome the impossibility result of Lemma \ref{lem:infinite}.
Additionally, we study an easier version of the problem, where we allow for a linear MSF in the number of items.
This second setting corresponds roughly to having all items within a fixed distance of their true rank.

Independently of the choice of an algorithm, we will derive the expected MSF obtained for a given number of users.
This will give us a lower bound on the number of queries needed to reach a certain expected MSF.
Indeed, because the user thresholds are unknown, we cannot ``select'' the users we want without asking them to provide at least one rating, because they are indistinguishable.
Finally, we will give insights on the tightness of our lower bound by studying the sample complexity of an algorithm that uses active learning to efficiently sort the items.

\section{Expected Maximum Spearman Footrule}
\label{sec:msf}

As stated in Section \ref{sec:model}, the user thresholds divide the items in a sequence of bins. 
This bin sequence naturally defines a partial order, so we can compute the associated MSF. 
Note that the bin sequence (and thus the MSF) is fully determined by the set of item scores and user thresholds.
Consequently, the MSF is a random variable which is independent of the choice of an algorithm (provided that the algorithms extract all the possible information from the users).
In what follows, we denote by $F$ this random variable. 
In this section, we compute $\E[\msf]$, the expected MSF  given the number of users $\nbi$  and the number of items $\nbu$,
where the expectation is taken with respect to the randomness of the item scores and the user thresholds. 
Theorems \ref{thm:msf_high} and \ref{thm:msf_low}, which are our main results, respectively give the expected MSF when the number of users and items are of the same order and when we have asymptotically more users than items.
We give ideas of proof in Section \ref{sec:proofs} and the full proofs in Appendix \ref{app:sec:main}.

\subsection{Main Theorems}
\label{sec:main_theorems}

Recall that $\nbi$ is the number of items, $\nbu$ is the number of users, and $f_X$ and $f_Y$ are their densities.
Let $Y$ be a threshold selected uniformly at random (\emph{i.e.} $Y$ also has density $f_Y$).

\begin{theorem}
\label{thm:msf_high}
Assume that there exists $r \in \R^+ \emph{ s.t. } \nbu \sim r \nbi$ as $\nbi$ goes to infinity, then
$$
\left| \E[F] - \nbi \left(\frac{1}{2} +  \frac{1}{r}  \left[\frac{f_X(Y)^2}{f_Y(Y)^2}\right] \right) \right|
\leq \frac{r}{2} \nbi + o(\nbi) 
$$
\end{theorem}

\begin{theorem}
\label{thm:msf_low}

Assume that there exists $ r \in \R^+, \gamma > 1 $ \emph{ s.t. } $\nbu \sim r \nbi^{\gamma}$ as $n$ goes to infinity, then

$$
\E[F]   
\sim \frac{2}{r} n^{2-\gamma} \E \left[\frac{f_X(Y)^2}{f_Y(Y)^2} \right] 
$$
\end{theorem}

\paragraph{Interpretation of the theorems} 
Theorem \ref{thm:msf_high} shows that, if the number of items grows linearly with the number of users, then the expected MSF also grows linearly.
Theorem \ref{thm:msf_low} shows in particular that we need to take $\gamma=2$, \emph{i.e.} a quadratic number of users, if we want to keep a constant MSF.
If the user thresholds were known, an efficient algorithm could discard users that do not provide additional information without asking any rating from them. 
But, in our case, it is pointless to reject a user who has not provided any rating, because all users who have not rated anything are indistinguishable.
Consequently, the number of users constitutes a natural lower bound of the number of queries.

\paragraph{Influence of the distributions}
The divergence term $\E \left[\frac{f_X(Y)^2}{f_Y(Y)^2}\right]$ appears in both theorems. 
This expression captures the effect of having different distributions for the items and the thresholds.
Indeed, the MSF scales quadratically with the size of the bins (see Lemma \ref{lem:msf_bin}).
This means that, for a given number of items and users, the smallest MSF is achieved when the thresholds are equally spaced between the items.
Additionally, we can guess that having a high concentration of items in a region without thresholds will have a strong effect on the MSF, whereas the converse is not true.
This asymmetry appears in the mathematical expression, because we have $\E \left[\frac{f_X(Y)^2}{f_Y(Y)^2}\right] = \int_0^1 \frac{f_X(y)^2}{f_Y(y)} dy$, which shows that this term will blow up if for instance the support of $f_Y$ is smaller than the one of $f_X$.
In Appendix \ref{app:sec:divergence}, Lemma \ref{lem:min_div}, we show by convexity that this divergence is greater than $1$, the value $1$ being reached for $f_X=f_Y$. 
In addition, we show in Lemma \ref{lem:divergence_beta} that in the case where scores follow a $Beta(a_X, b_X)$, and thresholds a $Beta(a_Y, b_Y)$, we have the closed form 
$\E \left[\frac{f_X(Y)^2}{f_Y(Y)^2}\right] 
=\frac{\betaf(a_Y, b_Y)}{\betaf(a_X, b_X)^2}  \betaf(2a_X - a_Y, 2b_X - b_Y) $,
where $\betaf$ is the Beta function.

\subsection{Proof of Theorems}
\label{sec:proofs}

We present the principal lemmas used for Theorems \ref{thm:msf_high} and \ref{thm:msf_low}.
Full proofs are deferred to Appendix \ref{app:sec:main} and \ref{app:sec:properties}.

Our model allows for general densities for both item scores and user thresholds. 
However, we can observe that composing all scores and thresholds by an increasing function on $[0,1]$ would leave the order of the items and thresholds unchanged.
This shows that there exists different pairs $(f_X, f_Y)$ that have the same properties with respect to our problem.
In particular, from any $(f_X, f_Y)$, we can compose everything by the \emph{cdf} $F_Y$, in order to have a uniform distribution of the thresholds on $[0,1]$.
Thus, any $(f_X, f_Y)$ model has an equivalent $(f_{X'}, 1)$ model, with the same properties when it comes to the MSF.
Therefore, in the rest of the section, we assume that $f_Y=1$, \emph{i.e.} the user thresholds are uniform \emph{iid} in $[0,1]$. 
Under this assumption, our results are expressed in function of $\E[f_X(Y)^2]$.
We can show that this term becomes the $\E\left[ \frac{f_X(Y)^2}{f_Y(Y)^2} \right]$ of Theorems 1 and 2 for $f_Y \neq 1$.
We formally explain this manipulation in Appendix \ref{app:sec:reduction} and prove the full result in Lemma \ref{lem:rescale}.

We now present our main lemmas, proven under the assumption $f_Y=1$.
The results of Theorems \ref{thm:msf_high} and \ref{thm:msf_low} both come from the following lemma, which splits the expected MSF in two terms.

\begin{lemma}
\label{lem:msf_original}
Let $B$ be the number of items in a bin chosen uniformly at random among all the bins. 
Then,
$$
\E[F] = \frac{1}{2} (\nbu + 1) (\E[B^2] - \prob(B \isodd))
$$
\end{lemma}

Lemmas \ref{lem:msf_original}
is proven in Appendix \ref{app:sec:msf_original}, using Lemmas 
\ref{lem:msf_bin_sum} and \ref{lem:msf_bin},
which respectively show that the MSF of a bin sequence is the sum of the MSF of each of its bins, and that the MSF of a bin of $b$ items is $b^2/2$ if $b$ is even and $(b^2-1)/2$ if $b$ is odd.
It shows that the MSF can be computed from $\E[B^2]$ and $\prob(B \isodd)$, where $B$ is the size of a bin selected uniformly at random.
These two terms are computed in the following two lemmas, which are proven in Appendix \ref{app:sec:preliminary}. 

\begin{lemma}
\label{lem:exp_square_bin}
Let $\beta \in (0.5, 1)$. 
Then,
$$
\E[B^2]
=  \frac{\nbi}{\nbu+1} +  2  \frac{\nbi^2-\nbi}{(\nbu+1)^2} \E [  f_X(Y)^2]  +  O\left(\frac{\nbi^2}{\nbu^{2 \beta + 1}}  \right) 
$$
\end{lemma}

\begin{lemma}
\label{lem:proba_b_odd}
Let $\beta \in (0.5, 1)$, $\nbu = \omega(\nbi)$. 
Then, for $\nbi$ going to infinity,
$$
\prob(B \isodd) 
= \frac{\nbi}{\nbu} - 2\frac{\nbi^2}{\nbu^2}\E [  f_X(Y)^2]  
+ O \left(\frac{n}{m^{2\beta}} \right)  + O \left(\frac{n^3}{m^3}  \right) \\
$$
\end{lemma}

Lemmas \ref{lem:exp_square_bin} and \ref{lem:proba_b_odd} are results on the number of items in a random bin.
In order to prove them, we need to study the distribution of the number of items in a bin, which is hard to derive directly.
However, its distribution conditioned on the two extremities of the bin is easier to derive.
For all $k\in [\nbu]$, let $D_k \triangleq Y_{k+1} - Y_k$ be the \emph{length} of the bin (with $Y_0=0$ and $Y_{\nbu+1} = 1$).
We further define 
$
\forall k \in [\nbu],  \quad
P_k
\triangleq \prob(X_i \in [Y_k, Y_{k+1}] | Y_k, D_k) 
= \int_{Y_k}^{Y_k+D_k} f_X(x) dx 
$. 
The expression does not depend on $i$, because the item scores are \emph{iid}, and independent of the thresholds. 
Furthermore, the number of items in bin $k$ conditioned on $(Y_k, D_k)$ follows a binomial distribution of parameters $(\nbi, P_k)$.
This is because $P_k$ is the conditional probability that an item is in bin $k$.
In order to compute $\E[B^2]$ in Lemma \ref{lem:exp_square_bin}, we need the values of $\E[P_k^2]$ for all $k$.
We use Lemma \ref{lem:expectation_pk2} to approximate these values.

\begin{lemma}
\label{lem:expectation_pk2}
$\forall \beta \in (0.5, 1)$, when $\nbu$ goes to infinity,
$$
\E \left[\sum_{k=0}^\nbu P_k^2 \right]
= \sum_{k=0}^\nbu \E[(D_k f_X(Y_k))^2 ] + O \left(\frac{1}{\nbu^{2\beta}}\right) \\
$$
\end{lemma}

In order to prove Lemma \ref{lem:expectation_pk2}, we define a probabilistic event 
$\event(\beta, \nbu) \triangleq \left( \bigcap_{k=0}^\nbu \left( D_k \leq  \frac{1}{\nbu^\beta} \right) \right) 
$ (see Appendix \ref{app:sec:event}), whose probability goes exponentially to $1$ as $\nbu$ goes to infinity (Lemma \ref{lem:all_bins_variation}).
Then, we show that $\mathcal{E}(\beta, \nbu)$ implies that $P_k$ is well approximated by $D_k f_X(Y_k)$ (Lemma \ref{lem:proba_asymp}).
The remaining terms $\E[(D_k f_X(Y_k))^2 ]$ are computed in Appendix \ref{app:sec:properties}.

\section{Upper Bound On The Number of Queries}
\label{sec:algorithm}

In this section we provide some intuition on the tightness of our lower bound.
For this, we present an algorithm that solves the task of partial ranking from ratings, and we analyze its complexity in the case $f_X=f_Y$.
We provide experiments on the empirical complexity of our algorithm in Section \ref{sec:experiments}.

As explained in Section \ref{sec:model}, our model assumes that users arrive sequentially to the system, in an arbitrary order. Thus, for each user, nothing is known initially about their threshold, beyond their prior distribution. For each user, the algorithm adaptively chooses a sequence of items to query.
The algorithm can choose the items based on its current state, which depends on all the responses it has gathered till then (including the current user’s past responses).
The algorithm may choose any number of items to query the user with, and observes the corresponding noiseless ratings. 

The algorithm we propose maintains a bin sequence representing everything we know about the order of the items, which is updated with every new user. 
For each user $\idxuser$, it extracts all the information they can provide.
This means finding out which item scores are smaller than $\thresh_\idxuser$ and which are bigger.
In order to do so, we make the user rate all items of the bin that contains $\thresh_\idxuser$.
The position of the other item scores relative to $\thresh_\idxuser$ can be inferred because the bins are ordered in the bin sequence.

In order to find the bin containing $\thresh_\idxuser$ bin with the least number, we rely on the following observation: if the new user gives a rating of $1$ to an item of bin $\bin_\idxbin$, it implies that they would also give a rating of $1$ to all items of any bin $\bin_\idxbinb$, with $\idxbinb>\idxbin$.
So the index of the bin containing $\thresh_\idxuser$ has to be smaller or equal to $\idxbin$.
This shows that it is possible to perform a binary search on the bins, with a small subtlety:
for each bin selected during the binary search, we only learn if the index of the bin containing the threshold is \emph{smaller or equal} or \emph{greater or equal}.
Consequently, using only one query per bin, it is possible to isolate a pair of consecutive bins out of which one is the correct one, but it is not possible to know which of the two.

  Using this idea, we present our algorithm $\userdic$ (Threshold Binary Search), which is detailed in Algorithm \ref{alg:userdic}.
 The execution of $\userdic$ is divided in steps, each step corresponding to a different user.
  Each step is divided in three phases : $\bsearch$, $\bisolate$ and $\bsplit$.
  These phases are detailed 
  in Algorithms \ref{alg:bsearch}, \ref{alg:bisolate} and \ref{alg:bsplit}.
  In what follows, we refer by $\bsearch_u$, $\bisolate_u$, $\bsplit_u$ and $\bclear_u$ to the execution of these phases at step $u$.\\ 
 Let $\thresh_u$ be the (unknown) threshold of user $u$.
   $\bsearch_u$ corresponds to the binary search described before. It finds a subset of two adjacent bins of which one of the two contains $\thresh_u$ (or $\thresh_u$ is between the two bins).
  $\bisolate_u$ identifies which of these two bins is the one actually containing $\thresh_u$ by alternating the queries between the two bins (if $\thresh_u$ is between the bins, the biggest one is returned).
  $\bsplit_u$ splits the selected bin in two new bins by requesting the user to rate all items of the bin.
Finally, we update the bin sequence by replacing the bin being split by the two new bins.
If one of the two bins returned by $\bsplit_u$ is empty, it means that the new user does not bring additional information.
In this case, the bin sequence is not updated.

\paragraph{Upper bound on the complexity of $\userdic$}
We do not formally prove the complexity of our algorithm, but for the case where the densities of the scores and thresholds are the same and by making some reasonable asymptotic approximations, we are able to show that the expected number of queries needed is $O( \nbi \log(\nbu) + \nbu \log(\nbi))$.

\paragraph{Idea of proof} We split the total number of queries $\sampcomp$ in $\sampcomp = \tsearch + \tsplit + \tiso$, the cost of each of the three phases of the algorithm.
For each user, the $\bsearch$ phase cannot make more than $\log_2(\nbi)$ queries, because it is a binary search on the sequence of nonempty bins.
So we have $\tsearch \leq \nbu \log_2(\nbi)$. 
By using the fact that $\bisolate$ defaults to the biggest of the two bins, we can show that $\tiso \leq \tsplit$. 
Estimating the complexity of $\tsplit$ is the most difficult part of the analysis. 
To do it, for each user-item pair $(\idxitem, \idxuser)$, we look at the probability that the $u$-th user treated by the algorithm rates $\idxitem$ during the $\bsplit$ phase.
We show that this probability is roughly $\frac{1}{\idxuser}$, which gives a total cost of $\E[\tsplit] \simeq \sum_{\idxitem=1}^\nbi \sum_{\idxuser=1}^\nbu \frac{1}{\idxuser} \simeq \nbi \log(\nbu)$.
Combining these three results gives us the upper bound $\E[\sampcomp] = O(\nbi \log(\nbu) + \nbu \log(\nbi))$. 
We present the detailed derivations in Appendix \ref{app:sec:algo}.

\paragraph{Interpretation} In our cases of interest, where we have at least as many users as items, the dominating term is $\nbu \log(\nbi)$, \emph{i.e.} a logarithmic number of queries per user.
So, in these regimes, our lower bound (which relies on the number of users $m$) is optimal up to a logarithmic factor.

 \begin{algorithm}[htbp]
 \caption{ $\userdic(\nbi, \nbu)$
 }
 \label{alg:userdic}
 $\bin_1 \gets [\nbi]$ 
 \Comment*[r]{bin}
 $\binseq \gets (\bin_1)$ 
 \Comment*[r]{bin sequence}
 \For{$\idxuser$ going from $1$ to $\nbu$}{
     $ \idxleft, \idxright \gets \bsearch(\binseq, u)$ 
    \Comment*[r]{find a pair of bins containing $Y_\idxuser$ (Algo \ref{alg:bsearch})} 
     $k^* \gets \bisolate (\binseq, \idxleft, \idxright, u)$ 
    \Comment*[r]{identify the correct one (Algo \ref{alg:bisolate})} 
     $ \binminus, \binplus \gets \bsplit(\bin_{k^*}, u)$ 
    \Comment*[r]{try to split the bin in two (Algo \ref{alg:bsplit})} 
     \If{$|\binminus| > 0$ and $|\binplus| > 0$  }
     {
     \hspace{-10pt} $\binseq \hspace{-2pt} \gets \hspace{-3pt} (\bin_1, ..., \bin_{\idxbin^*-1}, \binminus, \binplus, \bin_{\idxbin^*+1}, ..., \bin_{|\binseq|})$ 
     }
     }
  \textbf{Return} $\binseq$
 \end{algorithm}

 \begin{algorithm}[htbp]
 \caption{$\bsearch(\binseq, u)$ \\
 Find a subset of two adjacent bins of which one contains the threshold.
 }
 \label{alg:bsearch}
   $\idxleft \gets 1$ \\
   $\idxright \gets |\binseq|$ \\
   \While{$\idxleft<\idxright -1$}{ 
       $\idxbin \gets \lfloor \frac{\idxleft+\idxright}{2} \rfloor$ \\
       $\idxitem \gets $ uniformly sampled item from $\bin_{\idxbin}$\\ 
       $\request$ rating $\rating{u}{\idxitem}$ \\
       {\eIf{$\rating{u}{\idxitem} = 1$}
         { $\idxright \gets \idxbin$}
         {$\idxleft\gets \idxbin$}
       } 
   } 
   \textbf{Return} $\idxleft, \idxright$
 \end{algorithm}

 \begin{algorithm}[htbp]
 \caption{ $\bisolate (\binseq, \idxleft, \idxright, \idxuser)$  \\ 
 Find which bin out of $\bin_\idxleft$ and $\bin_\idxright$ is the one containing $\thresh_u$.}
 \label{alg:bisolate}
 $R_\idxleft \gets \emptyset$ \\
 $R_\idxright \gets \emptyset$ \\
 \While{$\bin_\idxleft$ and $\bin_\idxright$  both contain an item that has not been rated by user $u$}{
     $\idxitem \gets$ uniformly sampled  item from $\bin_\idxleft \backslash R_\idxleft$  \\
     $\idxitemb \gets$ uniformly sampled item from $\bin_\idxright \backslash R_\idxright$  \\
     $R_\idxleft \gets R_\idxleft \cup \{\idxitem\}$ \\
     $R_\idxright \gets R_\idxright \cup \{\idxright\}$ \\
     $\request$ ratings $\rating{u}{\idxitem}$ and $\rating{u}{\idxitemb}$ \\
     \If{$\rating{u}{\idxitem} = 1$}
           {\textbf{Return} $\idxleft$}
     \If{$\rating{u}{\idxitemb} = 0$}
           {\textbf{Return} $\idxright$}
 }
 \eIf{$|\bin_\idxleft| > |\bin_\idxright|$}
     {\textbf{Return} $\idxleft$ }
     {\textbf{Return} $\idxright$}
 \end{algorithm}

 \begin{algorithm}[htbp]
 \caption{$\bsplit (\bin, \idxuser)$ \\ 
 Split the selected bin.
 }
 \label{alg:bsplit}
     \For{$\idxitem$ in $\bin$}
         {$\request$ rating $\rating{\idxuser}{\idxitem} $} 
     $\binminus \gets \{\idxitem \in \bin \mid \rating{\idxuser}{\idxitem} = 0\}$ \\
     $\binplus \gets \{\idxitem \in \bin \mid \rating{\idxuser}{\idxitem} = 1\}$  \\
     \textbf{Return} $\binminus, \binplus$
 \end{algorithm}

\newpage

\section{Experiments}
\label{sec:experiments}

We perform experiments\footnote{Our code is available at \url{https://anonymous.4open.science/r/threshold_model_neurips25}}  to support our claim of the previous section and show that the convergence of Theorems  \ref{thm:msf_high} and \ref{thm:msf_low} is fast enough. 

\subsection{Experiment setting}

We run Monte-Carlo experiments using Python.
All experiments are run 1000 times and we report the average results on the plots along with the 95\% confidence intervals.
We fix the randomness using seeds 1 to 1000 for reproducibility.
All experiments were ran using an Apple M3 Pro chip with 18GB of RAM.
Our experiments are constructed as follows.
First, we generate our data by sampling $\nbi$ \emph{iid} item scores from a $Beta(a_X, b_X)$ and $\nbu$ \emph{iid} user thresholds following a $Beta(a_Y, b_Y)$.
Second, we run $\userdic$ on this data until the algorithm terminates.
Finally, we plot our metrics (MSF obtained and number of queries made).
We display the average taken over the different runs, along with the 95\% confidence interval
(\emph{e.g.}
$\hat{\E}[\msf] \pm 1.96 \frac{\sqrt{\hat{\V}(\msf)}}{\sqrt{s}}$,
where $s$ is the number of samples, and $\hat{\E}$ and $\hat{\V}$ are the empirical mean and variance).
The MSF is computed by taking the bin sequence returned by the algorithm and using Lemmas \ref{lem:msf_bin_sum} and \ref{lem:msf_bin}.

\subsection{Experiment results}

We present the results when the number of users grows linearly with respect to the number of items on Figure \ref{fig:linear} and results for a quadratic number of users on Figure \ref{fig:quadratic}.

 \begin{figure*}[ht!]
  \vspace{-5pt}
 \centering
   \begin{subfigure}{.50\textwidth}
   \centering
   \includegraphics[width=\linewidth]{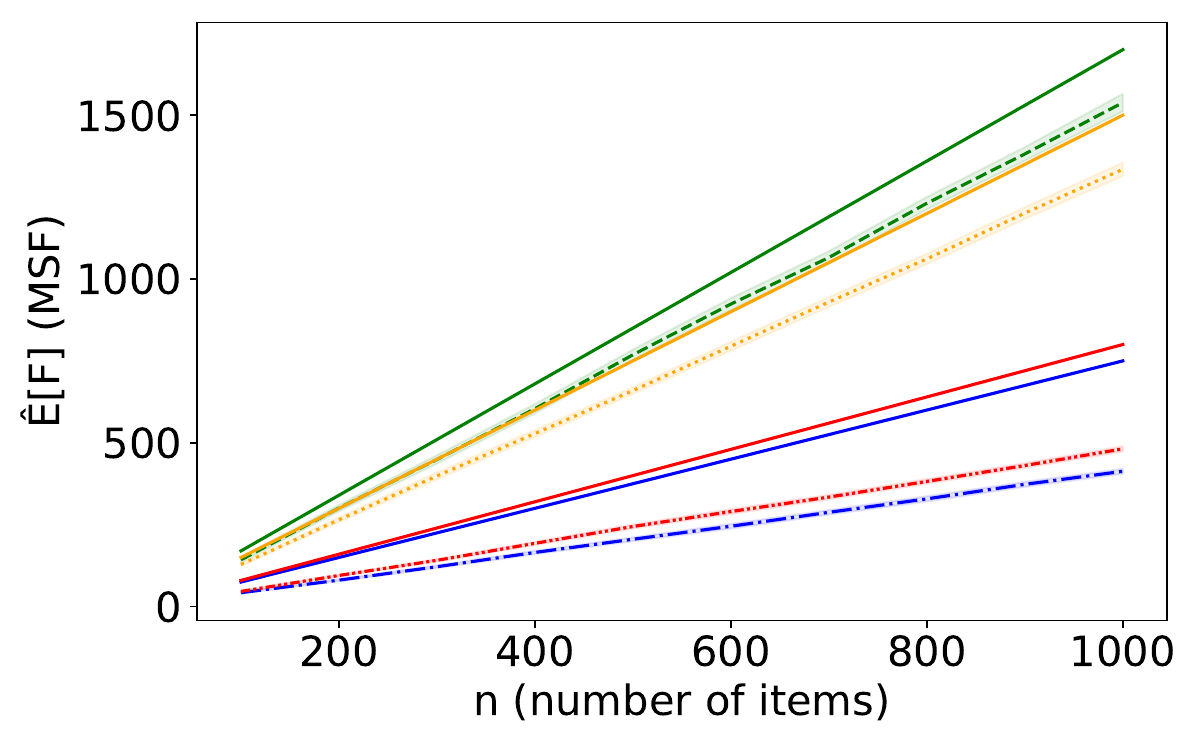}
  \caption{Expected MSF}
 \label{fig:expectedmsf_linear}
 \end{subfigure}\hfill
\begin{subfigure}{.50\textwidth}
  \centering
  \includegraphics[width=\linewidth]{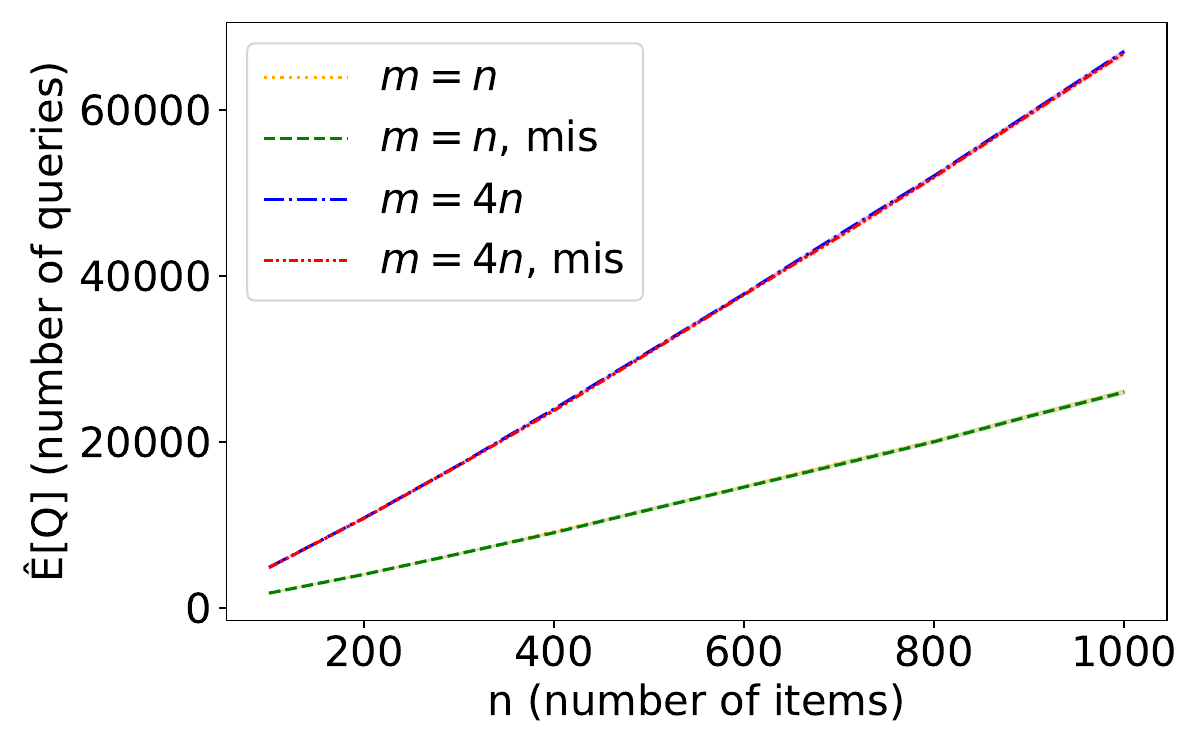}
  \caption{Expected number of queries}
 \label{fig:expectedqueries_linear}
\end{subfigure}
 \caption{Experiments on the expected MSF and number of queries for a number of users linear in the number of items.
The full lines on Figure (a) are the theoretical values of Theorem \ref{thm:msf_high}.
 Figure (b) shows the empirical query cost of $\userdic$.
The label \emph{'mis'} indicates the experiments with a mismatch between the distributions of the items and the thresholds.
 We use $a_X=2, b_X=3, a_Y=2, b_Y=2$ for the mismatch case, and $a_X=1, b_X=1, a_Y=1, b_Y=1$ in the default case.
 On Figure (b), confidence intervals are too small to display, and lines for the mismatch case overlap with their counterparts for the matching case.
 }
 \label{fig:linear}
 \vspace{-12pt}
 \end{figure*}

 \begin{figure*}[ht!]
  \vspace{-5pt}
 \centering
   \begin{subfigure}{.50\textwidth}
   \centering
   \includegraphics[width=\linewidth]{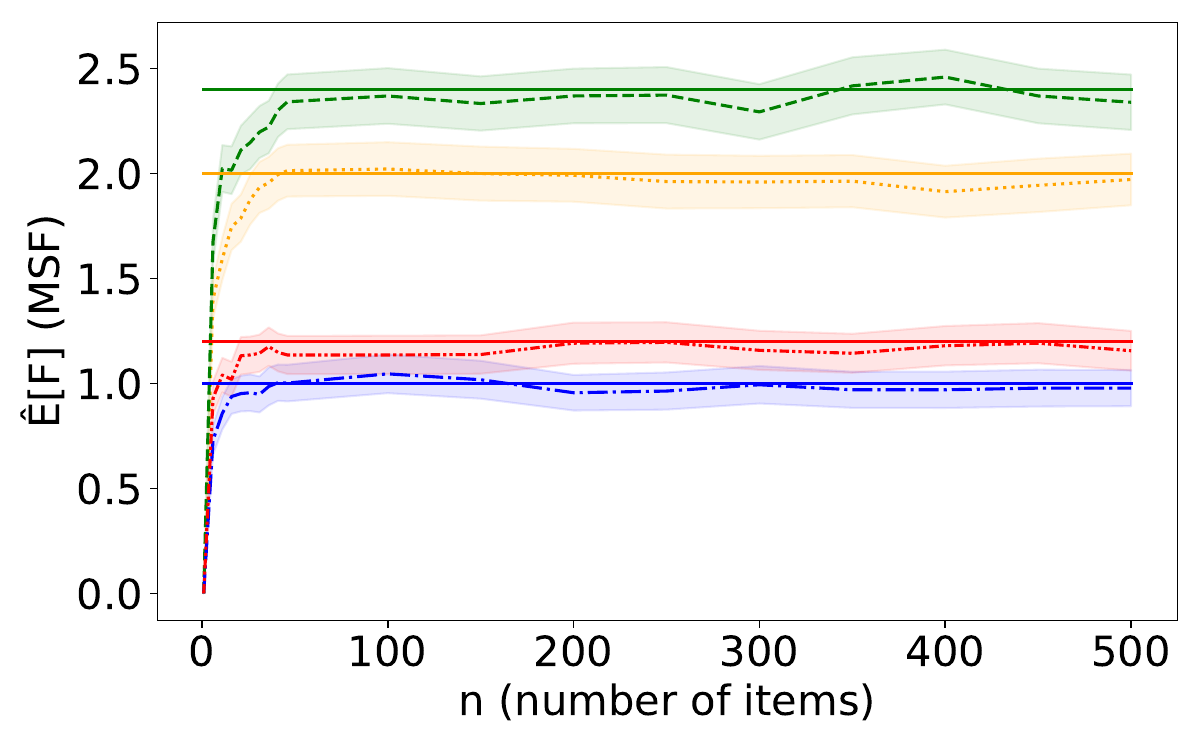}
  \caption{Expected MSF}
 \label{fig:msf_quadratic}
 \end{subfigure}\hfill
\begin{subfigure}{.50\textwidth}
  \centering
  \includegraphics[width=\linewidth]{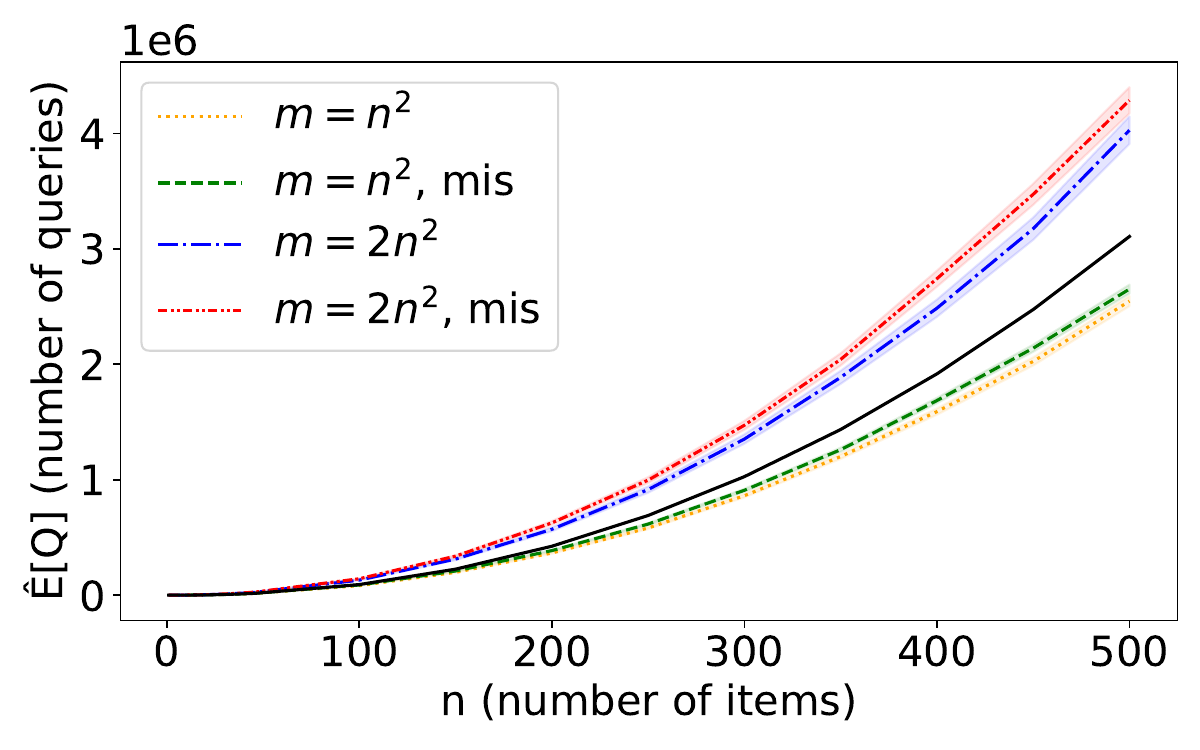}
  \caption{Expected number of queries}
 \label{fig:queries_quadratic}
\end{subfigure}
 \caption{Experiments on the expected MSF and number of queries for a number of users quadratic in the number of items.
The full lines on Figure (a) are the theoretical values of Theorem \ref{thm:msf_low}.
 Figure (b) shows the empirical query cost of $\userdic$.
 The solid black line shows $y = 2 n^2 \log(n)$, which corresponds to the rate estimated in Section \ref{sec:algorithm} (the constant $2$ is manually chosen to roughly match the other lines). 
The label \emph{'mis'} indicates the experiments with a mismatch between the distributions of the items and the thresholds.
 We use $a_X=2, b_X=3, a_Y=2, b_Y=2$ for the mismatch case, and $a_X=1, b_X=1, a_Y=1, b_Y=1$ in the default case.
 }
 \label{fig:quadratic}
 \vspace{-12pt}
 \end{figure*}

\subsection{Interpretation of results}

\paragraph{Linear regime}
We can see on Figure \ref{fig:expectedmsf_linear} that the expected MSF is linear in the number of items, as stated by Theorem \ref{thm:msf_high}.
Additionally, we observe that the MSF is smaller than the estimation of  
$n(\frac{1}{2} +  \frac{1}{r}  (f_X(Y) /f_Y(Y))^2)$
of Theorem \ref{thm:msf_high}, with the gap growing linearly.
Theorem \ref{thm:msf_high} showed that we have  
$\E[F] = n( \frac{1}{2} +  \frac{1}{r}  (f_X(Y) /f_Y(Y))^2) + g(n) n + o(n)
$
with $|g| \leq \frac{r}{2}$.
According to Figure \ref{fig:expectedmsf_linear} we most likely have $g<0$.
We can see on Figure \ref{fig:expectedqueries_linear} that the number of queries of $\userdic$ follows the $\nbu \log(\nbi)$ tendency discussed in Section \ref{sec:algorithm} ($n\log(n)$ here), which confirms the analysis of Lemmas \ref{lem:qsearchpm}, \ref{lem:tisoleqtsplit} and \ref{lem:qsplit}.
We also observe that the mismatch in the linear regime has very little effect on the number of queries per user. 
This indicates that even though our upper bound analysis is limited to equal distributions for the items and thresholds, the approximate upper bound we provide is still reasonable in the mismatched case.

\paragraph{Quadratic regime}

We see on Figure \ref{fig:msf_quadratic} that the empirical MSF reaches the limit predicted by Theorem \ref{thm:msf_low} at around 50 items.
This shows that our asymptotic results can be used to estimate the expected MSF even for a small number of items. 
We observe on Figure \ref{fig:queries_quadratic} that the number of queries of $\userdic$ follows the $\nbu \log(\nbi)$ (\emph{i.e.} $\nbi^2 \log(\nbi)$ here) tendency discussed in Section \ref{sec:algorithm}.
Interestingly, we can also see that the mismatch increases the average number of queries per user, but only by a small margin. 
These results support our theoretical analysis of Sections \ref{sec:msf} and \ref{sec:algorithm} by providing strong evidence that our bounds are tight up to a logarithmic factor.

\section{Discussion and Related Work}
\label{sec:conclusion}

This work addresses the problem of adaptive ranking from discrete ratings.
Prior work on adaptive ranking has primarily relied on pairwise comparisons to order items \citep{jamieson2011active, heckel2019active}, with some studies highlighting the efficiency gains of settling for approximate rankings \citep{heckel2018approximate}.
Classical sorting algorithms such as QuickSort also rely on comparisons. A wide range of comparison oracle models have been proposed, but none interpret comparisons as being derived from discrete ratings.
Our model introduces a new kind of oracle: a comparison is only available when a user’s threshold separates the scores of the two items. We show that ranking items under this oracle is significantly harder than even a noisy comparison oracle like Plackett-Luce.

A closely related work is \citet{garg2019designing}, which also tackles the problem of inferring a ranking from discrete ratings.
Like we do for our algorithm, they emphasize the importance of posing the ``right question''.
However, their model differs significantly: they assume a homogeneous user population that provides noisy ratings based on item utilities, and their focus is on selecting the best query to pose to the entire population.
In contrast, we consider a heterogeneous population and focus on identifying the right user whose threshold yields informative ratings.
\citet{medo2010effect} also explore the effects of discretization in recommender systems.
Although their model shares structural elements with ours, {\em e.g.}, items with latent scores on a unit interval and ratings derived via thresholding—they assume uniform thresholds and focus on algorithmic consequences within collaborative filtering.
In contrast, we explore how unknown, user-specific thresholds complicate the problem of ranking items even when users agree on the ground-truth ordering. 

More broadly, our work is motivated by a long-standing interest in how best to elicit feedback in recommender and crowdsourcing systems. This question has appeared under different names, including {\em type of feedback} \citep{babski2023inferring}, {\em method of elicitation} \citep{shah2016estimation}, and {\em format} \citep{fernandes2023bridging}. A central debate concerns the trade-offs between cardinal (rating-based) and ordinal (comparison-based) feedback. Ratings are easier to aggregate across users and are the default in many systems. However, growing empirical evidence supports the superiority of comparisons in several respects: they tend to be less biased \citep{fernandes2023bridging}, more consistent over time \citep{jones_improving_2011}, and cognitively less demanding \citep{shah2016estimation, xu2024perceptual}. Studies have shown that user rating biases are pervasive and prone to drift over time \citep{harik2009examination}, complicating aggregation. Although some works \citep{wang2019your} suggest that cardinal feedback can outperform ordinal feedback, these rely on continuous rating scales—a setting fundamentally different from ours.
Our results complement these findings, and should be interpreted in the light of works like the ones of \cite{sparling2011rating} and \cite{jones_improving_2011}, who study the time taken, cognitive load or user satisfaction when using different types of feedbacks.

By isolating the cost of discretization in ratings,
our work provides new insights into the debate between of ordinal and cardinal feedback.
Indeed, we show that even in the absence of noise, the cost incurred by the discretization is leading to a high complexity for ratings, higher than the one of comparisons.
We showed that even if we had a priori knowledge that a group of users agree on the underlying item ranking, extracting that ranking efficiently is costly because we need to expend queries on every new user to learn about her unknown threshold, {\em before} we can extract useful information from that user to contribute to the estimated ranking.
Although we made a number of simplifying assumptions in this model, we believe that the broader observation on rating feedback is valuable: the fact that users have diverse thresholds necessitates spending a large number of queries to determine whether they can indeed order items.

\paragraph{Main takeways}
The most important implication of our work is that when fine-grained ranking is required, \textit{comparison-based feedback} is significantly more sample-efficient than \textit{discrete ratings}. As we show, $O(n\log n)$ rating queries suffice to estimate ranks within a constant deviation per item (linear MSF), but reducing this deviation further requires a much larger number of queries — due to diminishing marginal information gain from each new user. Our work provides one interpretation of the commonly held belief: soliciting comparisons instead of ratings is better for ranking items. 
A second practical takeaway of this work is that one suffers a penalty when the distribution of item scores is mismatched with the distribution of user thresholds.
In this work, we have quantified this penalty via the term $\E [ (f_{X}(Y)/ f_{Y}(Y))^2 ]$, which grows larger as the two distributions become more different.
In practical systems, this suggests that \textit{knowing the score distribution} and designing feedback prompts to make users most discriminative around that region is crucial.
This insight complements the work of Garg and Johari (2019), who also study the design of optimal binary prompts for ranking (e.g., asking “Is this essay better than average?” versus “Is it better than good?”).

\paragraph{Limitations}
Our model assumes that items have a fixed utility shared across users and that the only source of randomness is the distribution of discretization thresholds.
Relaxing either assumption, for example, by incorporating noise in ratings or allowing for user disagreements, would likely make the ranking problem harder.
 As such, our lower bounds should still hold in broader settings, albeit possibly in weaker forms.
The restriction to binary ratings is not fundamental.
 Our lower bounds generalize to $k$-level rating scales.
Indeed, if we make the assumption that each of the $\nbu$ users has $k-1$ \emph{iid} thresholds in $[0,1]$ instead of only one threshold,
the distribution of the set of all thresholds is the same as having $\nbu(k-1)$ users with one threshold each. 
In this case, the lower bound on the number of users needed scales down by a factor of $k-1$.
Finally, our algorithm is tightly coupled to our model assumptions and may not perform well under real-world noise or preference heterogeneity.
Its primary role is to match our lower bounds and establish tightness.
That said, the algorithm's core idea—using binary search to isolate sets of items for a user to rate—offers a practical design principle that may generalize to more robust algorithms in future work.


%
%

\bibliographystyle{abbrvnat}
\bibliography{arxiv_main}

\begin{thebibliography}{13}
\providecommand{\natexlab}[1]{#1}
\providecommand{\url}[1]{\texttt{#1}}
\expandafter\ifx\csname urlstyle\endcsname\relax
  \providecommand{\doi}[1]{doi: #1}\else
  \providecommand{\doi}{doi: \begingroup \urlstyle{rm}\Url}\fi

\bibitem[Babski-Reeves et~al.(2023)Babski-Reeves, Eksioglu, and Hampton]{babski2023inferring}
K.~Babski-Reeves, B.~Eksioglu, and D.~Hampton.
\newblock Inferring user preferences using cardinal vs. ordinal feedback in recommender systems.
\newblock \emph{IISE Annual Conference.Proceedings}, 2023.

\bibitem[Fernandes et~al.(2023)Fernandes, Madaan, Liu, Farinhas, Martins, Bertsch, de~Souza, Zhou, Wu, Neubig, et~al.]{fernandes2023bridging}
P.~Fernandes, A.~Madaan, E.~Liu, A.~Farinhas, P.~H. Martins, A.~Bertsch, J.~G. de~Souza, S.~Zhou, T.~Wu, G.~Neubig, et~al.
\newblock Bridging the gap: A survey on integrating (human) feedback for natural language generation.
\newblock \emph{Transactions of the Association for Computational Linguistics}, 11:\penalty0 1643--1668, 2023.

\bibitem[Garg and Johari(2019)]{garg2019designing}
N.~Garg and R.~Johari.
\newblock Designing optimal binary rating systems.
\newblock In K.~Chaudhuri and M.~Sugiyama, editors, \emph{Proceedings of the Twenty-Second International Conference on Artificial Intelligence and Statistics}, volume~89 of \emph{Proceedings of Machine Learning Research}, pages 1930--1939. PMLR, 2019.

\bibitem[Harik et~al.(2009)Harik, Clauser, Grabovsky, Nungester, Swanson, and Nandakumar]{harik2009examination}
P.~Harik, B.~E. Clauser, I.~Grabovsky, R.~J. Nungester, D.~Swanson, and R.~Nandakumar.
\newblock An examination of rater drift within a generalizability theory framework.
\newblock \emph{Journal of Educational Measurement}, 46\penalty0 (1):\penalty0 43--58, 2009.

\bibitem[Heckel et~al.(2018)Heckel, Simchowitz, Ramchandran, and Wainwright]{heckel2018approximate}
R.~Heckel, M.~Simchowitz, K.~Ramchandran, and M.~Wainwright.
\newblock Approximate ranking from pairwise comparisons.
\newblock In A.~Storkey and F.~Perez-Cruz, editors, \emph{Proceedings of the Twenty-First International Conference on Artificial Intelligence and Statistics}, volume~84 of \emph{Proceedings of Machine Learning Research}, pages 1057--1066. PMLR, 2018.

\bibitem[Heckel et~al.(2019)Heckel, Shah, Ramchandran, and Wainwright]{heckel2019active}
R.~Heckel, N.~B. Shah, K.~Ramchandran, and M.~J. Wainwright.
\newblock Active ranking from pairwise comparisons and when parametric assumptions do not help.
\newblock \emph{The Annals of Statistics}, 2019.

\bibitem[Jamieson and Nowak(2011)]{jamieson2011active}
K.~G. Jamieson and R.~Nowak.
\newblock Active ranking using pairwise comparisons.
\newblock \emph{Advances in neural information processing systems}, 24, 2011.

\bibitem[Jones et~al.(2011)Jones, Brun, and Boyer]{jones_improving_2011}
N.~Jones, A.~Brun, and A.~Boyer.
\newblock Improving reliability of user preferences: {Comparing} instead of rating.
\newblock In \emph{2011 {Sixth} {International} {Conference} on {Digital} {Information} {Management}}, pages 316--321, Sept. 2011.

\bibitem[Medo and Wakeling(2010)]{medo2010effect}
M.~Medo and J.~R. Wakeling.
\newblock The effect of discrete vs. continuous-valued ratings on reputation and ranking systems.
\newblock \emph{Europhysics Letters}, 91\penalty0 (4):\penalty0 48004, 2010.

\bibitem[Shah et~al.(2016)Shah, Balakrishnan, Bradley, Parekh, Ramch, Wainwright, et~al.]{shah2016estimation}
N.~B. Shah, S.~Balakrishnan, J.~Bradley, A.~Parekh, K.~Ramch, M.~J. Wainwright, et~al.
\newblock Estimation from pairwise comparisons: Sharp minimax bounds with topology dependence.
\newblock \emph{Journal of Machine Learning Research}, 17\penalty0 (58):\penalty0 1--47, 2016.

\bibitem[Sparling and Sen(2011)]{sparling2011rating}
E.~I. Sparling and S.~Sen.
\newblock Rating: how difficult is it?
\newblock In \emph{Proceedings of the fifth ACM conference on Recommender systems}, pages 149--156, 2011.

\bibitem[Wang and Shah(2019)]{wang2019your}
J.~Wang and N.~B. Shah.
\newblock Your 2 is my 1, your 3 is my 9: Handling arbitrary miscalibrations in ratings.
\newblock In \emph{Proceedings of the 18th International Conference on Autonomous Agents and MultiAgent Systems}, pages 864--872, 2019.

\bibitem[Xu et~al.(2024)Xu, McRae, Wang, Davenport, and Pananjady]{xu2024perceptual}
A.~Xu, A.~McRae, J.~Wang, M.~Davenport, and A.~Pananjady.
\newblock Perceptual adjustment queries and an inverted measurement paradigm for low-rank metric learning.
\newblock \emph{Advances in Neural Information Processing Systems}, 36, 2024.

\end{thebibliography}

\medskip


\newpage
\newpage
\appendix

\section{Outline}

We provide here some intuitions and pointers to help the reader navigate through the Appendix.

\subsection{On the expected MSF}

\paragraph{Scaling of the MSF} Our main results, Theorems \ref{thm:msf_high} and \ref{thm:msf_low}, are proven in Appendix \ref{app:sec:main}.
In essence, they show that the MSF scales as $\frac{n}{m^2} $, where $n$ is the number of items and $m$ the number of users.
Both theorems rely on Lemma \ref{lem:msf_original}, which shows that the MSF grows quadratically with the bin sizes (number of items in each bin), and linearly with the number of bins (\emph{i.e.} the number of users plus one).
Lemma \ref{lem:msf_original} also 
 shows that the parity of the bin sizes has a small bounded effect on the MSF. 
 This effect is a consequence of the dichotomy of Lemma \ref{lem:msf_bin}.
In practice, this means that the parity of the bin size has a significant effect only when the MSF is not going to infinity.
For this reason, deriving the probability that bins contain an odd number of items is not needed for Theorem \ref{thm:msf_high}, which deals with high MSF regimes. 
On the contrary, Theorem \ref{thm:msf_low} deals with the cases where the MSF is small. 
Therefore, its proof requires additional analysis to take into account the probability of bins sizes being odd.

\paragraph{Properties of the bins}
The number of items in the $k$-th bin $B_k$ depends on both the length of the bin $D_k$ and the density of the items $f_X$ on the $[Y_k, Y_{k+1}]$ interval. 
As the number of thresholds $m$ goes to infinity, we can use the Lipschitzness of $f_X$ to argue that $f_X$ is roughly constant of value $f_X(Y_k)$ over the bin. 
This idea is formalized in Appendix \ref{app:sec:event}.
Because they are the differences of the order statistics of \emph{iid} independent random variables, the bin lengths $(D_k)_k$ follow Beta distributions (Lemma \ref{lem:diffunif}).
Although the $(D_k)_k$ are not independent, they essentially converge to \emph{iid} exponential random variables as $m$ goes to infinity.
However, for finite $m$ analysis, we need to take into account not only the dependency between the $(D_k)_k$, but also the dependency between the $(D_k)_k$ and the $(Y_k)_k$.
For this reason, we work with conditional distributions.
In particular, we prove several results on the length $D$ of a bin selected uniformly at random conditional on the left extremity $Y$ of the bin (Lemmas \ref{lem:d_cond_y}, \ref{lem:d1_cond_y}, \ref{lem:d2_cond_y}).

\paragraph{Border effects}
Although the results would be the same by overlooking it, rigorous analysis has to treat differently the first bin, which has a deterministic threshold on its left side ($Y_0 = 0$). 
Indeed, when considering the set of all bins $\{B_k| k \in [m] \cup \{0\} \}$ and their left extremities $\{Y_k| k \in [m] \cup \{0\} \}$, the first threshold $Y_0$ is not an order statistic of \emph{iid} random variables.
This means that this special case must be treated separately, which makes an additional term appear in the proofs, although it turns out to be negligible.
Because of this subtlety, $Y$ (a threshold selected uniformly at random) does not exactly have density $f_Y$, but rather converges in law to density $f_Y$ as $m$ goes to infinity.
Understanding this can help the reader process the proofs with more ease.

\subsection{On the complexity of the $\userdic$ algorithm}

In Appendix \ref{app:sec:algo}, we study the complexity of our algorithm.
We recall here the main ideas, already presented in Section \ref{sec:algorithm}.
The  analysis relies on splitting the total number of queries $\sampcomp$ in $\sampcomp = \tsearch + \tsplit + \tiso$, the cost of each of the three phases of the algorithm.
For each user, the $\bsearch$ phase cannot make more than $\log_2(\nbi)$ queries, because it is a binary search on the sequence of nonempty bins.
So we have $\tsearch \leq \nbu \log_2(\nbi)$. 
By using the fact that $\bisolate$ defaults to the biggest of the two bins, we can show that $\tiso \leq \tsplit$. 
Estimating the complexity of $\tsplit$ is the most difficult part of the analysis. 
To do it, for each user-item pair $(\idxitem, \idxuser)$, we look at the probability that the $u$-th user treated by the algorithm rates $\idxitem$ during the $\bsplit$ phase.
We show that this probability is roughly $\frac{1}{\idxuser}$, which gives a total cost of $\E[\tsplit] \simeq \sum_{\idxitem=1}^\nbi \sum_{\idxuser=1}^\nbu \frac{1}{\idxuser} \simeq \nbi \log(\nbu)$.
Combining these three results gives us the upper bound $\E[\sampcomp] = O(\nbi \log(\nbu) + \nbu \log(\nbi))$.

\subsection{Structure of the Appendix}

In Appendix \ref{app:sec:definitions}, we recall our definitions and introduce new notation useful to our analysis.
In Appendix \ref{app:sec:main}, we prove our main results. 
In particular we prove some preliminary lemmas in Appendix \ref{app:sec:preliminary} before proving Theorem \ref{thm:msf_high} in Appendix \ref{app:sec:msf_high} and Theorem \ref{thm:msf_low} in Appendix \ref{app:sec:msf_low}.
In Appendix \ref{app:sec:properties}, we present various properties of the random variables of our model, which are used to prove Theorems \ref{thm:msf_high} and \ref{thm:msf_low}.
In Appendix \ref{app:sec:event}, we define a probabilistic event under which the lengths of all bins are uniformly upper bounded.
In Appendix \ref{app:sec:divergence}, we show why we can assume the thresholds follow a uniform distribution without loss of generality, and we present some results on the quadratic divergence.
In Appendix \ref{app:sec:algo}, we present the detailed idea of proof of the complexity of $\userdic$.
In Appendix \ref{app:sec:inequalities}, we provide a few lemmas on inequalities used in the Appendix.

\section{Definitions and Notation}
\label{app:sec:definitions}

\subsection{Main Notation}

We summarize here the detailed notation of the random variables used in our theoretical results.
Throughout the paper, we use the following notation:

\begin{itemize}
    \item $n$: number of items
    \item $m$: number of users
    \item $f_X$: the density of the item scores on $[0,1]$. Scores are \emph{iid} on the interval.
    \item $f_Y=1$: the density of the user thresholds on $[0,1]$.
    Thresholds are \emph{iid} on the interval.
    \item $\forall i \in [\nbi], X_i$ is the (unordered) score of item $i$. $X_i$s are \emph{iid} of density $f_X$.
    \item $\forall k \in [\nbu], Y_k$ is the $k$-th smallest  threshold, \emph{i.e.} the $k$-th order statistic of $\nbu$ \emph{iid} random variables of density $f_Y$ (with convention $Y_0 = 0$, and $Y_{k+1} =1$).
    \item $\forall k \in [\nbu] \cup \{0\}, D_k \triangleq Y_{k+1}-Y_k$ the length of bin number $k$.
    \item $\forall k \in [\nbu] \cup \{0\}, B_k  \triangleq |\{i \in [\nbi] | X_i \in [Y_k, Y_{k+1}] \} |$, the number of items in bin $k$. 
    \item $B \triangleq B_K, K \sim \mathcal{U}([\nbu] \cup \{0\})$, the number of items in a random bin, chosen uniformly at random among the bins. 
    We define in the same way $Y \triangleq Y_K$ and $D \triangleq D_K $, a random threshold and the length of a random bin.
    \item $F$: the (random) value of the MSF.
\end{itemize}

Note that the $f_Y=1$ assumption is made without loss of generality, as detailed in Appendix \ref{app:sec:reduction}.

\subsection{Definition of the MSF}

We recall here the definitions necessary to define the MSF.
First, we define the need the following preliminary notation:

\begin{itemize}
    \item  For any finite subset $\bin$ of $\N$, $\permutset_\bin$ is the set of possible orders of $\bin$.
    \item For any order $\sigma \in \permutset_\bin$ and for any $i \in \bin$, $\sigma(i)$ is the rank of item $i$ among the elements of $\bin$.
\end{itemize}

Using this notation, we can define bins, bin sequences, the Spearman footrule and the MSF.

\begin{repdefinition}{def:binseq}
    \textbf{Bin Sequence and Bins.} We call \textbf{bin sequence} an ordered partition $\binseq = (\bin_1, ..., \bin_{|\binseq|})$ of $[\nbi]$ that respects the order of the item scores.  \\
    Formally :   
    $$\forall \idxbin < \idxbin' \in [|\binseq|], \forall \idxitem \in \bin_\idxbin, \idxitem' \in \bin_{\idxbin'}, \score_\idxitem < \score_{\idxitem'}$$ 
   $\bin_1, \ldots, \bin_{|\binseq|}$ are called \textbf{bins} of items.
   They are sets of items whose scores belong to a certain interval.\\
     We note $\permutset_\binseq$ the partial order induced by this bin sequence, \emph{i.e.} $\permutset_\binseq$ is the set of orderings compatible with $\binseq$.
\end{repdefinition}

\begin{repdefinition}{def:sf}
\textbf{Spearman Footrule.} 
    $\forall \sigma, \sigma^* \in \permutset_\bin$,  $\sffunc(\sigma, \sigma^*) \triangleq \sum_{i \in \bin} |\sigma(i) - \sigma^*(i)|$
\end{repdefinition}

Given a partial order, the MSF measures the worst (\emph{i.e.} maximum) Spearman Footrule between two total orders compatible with the partial order:

\begin{repdefinition}{def:msf}
\textbf{Maximum Spearman Footrule (MSF)}.\\
Let $\bin \subseteq \N$ be  a finite set. 
Let $\permutset$ be a subset of $\permutset_\bin$. 
Then,  
    $$ \msffunc(\permutset) 
    \triangleq \max_{\sigma, \sigma' \in \permutset} \sffunc(\sigma, \sigma')$$
\end{repdefinition}

\subsection{Other Notation}

We define here more advanced notation which is used in the proofs.

\begin{definition}[$P_k$, $P$, $p(Y,D)$]
\label{def:pk}    
We define $P_k$ as the conditional probability that an item is in bin $k$:
$$
\forall i \in [n], \forall k \in [\nbu] \cup \{0\},  \quad
P_k
\triangleq \prob(X_i \in [Y_k, Y_{k+1}] | Y_k, D_k) 
= \int_{Y_k}^{Y_k+D_k} f_X(x) dx 
\triangleq p(Y_k, D_k)
$$ 

In the same way as for the other random variables depending on $k$, we define 
$
P \triangleq P_K
$,
where $K$ is uniform on $[m] \cup \{0\}$.

\end{definition}

\paragraph{Remark}
The expression does not depend on $i$, because the item scores are \emph{iid} and independent of the thresholds.

We provide here the definition of event  $\mathcal{E}$, which is discussed more in detail in Appendix \ref{app:sec:event}.

\begin{definition}[$\mathcal{E}(\beta, m)$]
\label{def:event}
For all $ \beta \in (0,1),  \nbu \in \N$, we define the event $\mathcal{E}(\beta, \nbu)$ as follows:

$$
\mathcal{E}(\beta, \nbu) 
\triangleq \left( \bigcap_{k=0}^\nbu \left( D_k \leq  \frac{1}{\nbu^\beta} \right) \right) 
$$
\end{definition}

\section{On the Difficulty of Perfect Ranking}
\label{app:sec:perfect}

\begin{replemma}{lem:infinite}
Let $\score_1, \ldots, \score_\nbi$ be iid item scores following a $\mathcal{U}([0,1])$ distribution.
    Let  $\nbur$ be the random number of users needed to obtain a total order.  Then we have
    $$ \E[\nbur] = \infty $$
\end{replemma}

\begin{proof}
As discussed in Section \ref{sec:proofs}, and proved in Appendix \ref{app:sec:reduction}, we can make the assumption that either $f_X=1$ or $f_Y=1$ without loss of generality.
Here, we assume $f_X=1$, \emph{i.e.} the item scores are uniformly distributed on $[0,1]$.\\
For all $i \in [\nbi]$, let us note $\scoreord{i}$ the $i$-th smallest score.
Let $U_{1,2}$ be the random variable of the number of users sampled before the arrival of a threshold in $[\scoreord{1}, \scoreord{2}]$. 
Clearly, $U_{1,2}$ is smaller than the total number of users sampled before having a threshold between every pair of consecutive items. 
So we have $\E[\nbur] \geq \E[U_{1,2}]$.
Let us show  $\E[U_{1,2}] = \infty$.
Let us define $c_Y \triangleq \max f_Y$.
We have
\begin{align*}
    \E[U_{1,2}]
    &= \E \left[ \E[U_{1,2} \mid \scoreord{1}, \scoreord{2}] \right] \\
    &= \E \left[ \frac{1}{\int_{\scoreord{1}}^{\scoreord{2}} f_Y(y) dy} \right]  \\
    &\geq \E \left[ \frac{1}{\int_{\scoreord{1}}^{\scoreord{2}} c_Y dy} \right] \\
    &= \frac{1}{c_Y} \E \left[ \frac{1}{\scoreord{2} - \scoreord{1}} \right]  \\
    &= \frac{1}{c_Y} \int_0^1 \frac{1}{x} \nbi (1-x )^{\nbi-1} dx \\
    &=  + \infty
\end{align*}

     The second equality comes from the fact that the thresholds are independent in $[0,1]$, so the number of users needed to obtain a threshold in $[\scoreord{1}, \scoreord{2}]$ follows a geometric random variable of parameter 
    $p=\int_{\scoreord{1}}^{\scoreord{2}} f_Y(y) dy$. \\
    
The last two equalities are because $\scoreord{2} - \scoreord{1}$ is the difference of consecutive order statistics of a uniform distribution, which follows a Beta distribution $Beta(1, \nbi)$ (Lemma \ref{lem:diffunif}), and finally the integral diverges at $0$.

\end{proof}

\paragraph{Discussion on Lemma \ref{lem:infinite}.}

The intuition is that $M$, the number of users needed until we get a perfect ordering, is a random variable whose distribution has a very heavy tail: it has a decay of $1/m^2$.
Equivalently, the complementary CDF (i.e. $\prob(M>m)$), has a decay of $1/m$.
It is easy to show that such a distribution has infinite mean; it follows from the fact that the series $\sum_m 1/m$ is not summable.

Digging a little deeper sheds more light on why the CDF of  has such a heavy tail. 
The fundamental result is that we need a user threshold between every two item scores in order to rank the items fully.
In the simple case when there are only two items, we need a user threshold between them.
Assume for the sake of this argument that both item scores and user thresholds are uniformly distributed.
Let the two item scores be laid out first.
Given that the gap between the two item scores is  , the probability that a random user’s threshold lies between these two is $x$.
For each user, this event is independent. Therefore, the expected number of users to pass by until this random event happens is $1/x$.
Now, the gap between two items could be any number between zero and one. The exact distribution of the gap is a beta random variable, but the important point is that this distribution has strictly positive density at zero.
 Therefore, the probability of the gap being between  $x$ and $x+dx$ is proportional to $dx$.
Integrating  $c(1/x)dx$ from zero to one gives us infinity. A slightly misleading line of argument is the following.
When both the user thresholds and item scores are uniformly distributed, by symmetry, there is a $1/3$ chance that a user threshold lies between the item scores.
Since these events are independent, on average, it should take 3 users to land a user whose threshold lies between the score.
The fallacy in this argument is that the events: “the threshold of user $u$ lies between the two item scores” are not independent across users.
The fact that the first user’s threshold did not lie between the scores gives us some information that the item scores are probably closer together than what is known a priori.
This reduces the chance that the second user’s threshold lies between the scores.
We get independence of these events only by conditioning on the item scores. 
Arguably, this simple counter-intuitive fact is what makes this problem so interesting to study.

\section{Main Results}
\label{app:sec:main}

In this section, we prove our two main results, Theorems \ref{thm:msf_high} and \ref{thm:msf_low} in Appendix \ref{app:sec:msf_high} and \ref{app:sec:msf_low}, which give the expected MSF in function of the number of items and users.
Before this, we prove the main lemmas needed to prove our theorems in Appendix \ref{app:sec:preliminary}.

\subsection{Preliminary Results to the Theorems}
\label{app:sec:preliminary}

In this section, we show that the expected MSF depends on two terms: $\E[B^2]$ and $\prob(B \isodd)$. 
Then, we compute those two terms, thus dealing with most of the difficulties of Theorems \ref{thm:msf_high} and \ref{thm:msf_low}.

\subsubsection{Expressing the MSF in Function of $B$}
\label{app:sec:msf_original}

\begin{replemma}{lem:msf_original}
$$
\E[F] = \frac{1}{2} (\nbu +1) (\E[B^2] - \prob(B \isodd))
$$
\end{replemma}

\begin{proof}
\begin{align*}
\E[F]
&= \frac{1}{2} \sum_{\idxbin=0}^\nbu  
           (\prob(B_\idxbin \text{ is even}) \E[B_\idxbin^2 \mid B_\idxbin \text{ is even}]
 +  \prob(B_\idxbin \text{ is odd}) (\E[B_\idxbin^2 - 1 \mid B_\idxbin \text{ is odd}])) \\ 
&= \frac{1}{2} \sum_{\idxbin=0}^\nbu  \left( \E[B_\idxbin^2]  -  \prob(B_\idxbin \text{ is odd}) \right) \\
&= \frac{1}{2} (\nbu +1)  \left( \E[B^2]  -  \prob(B \text{ is odd}) \right) 
\end{align*}
where the first equality uses Lemmas \ref{lem:msf_bin_sum} and \ref{lem:msf_bin}.
 
\end{proof}

Lemma \ref{lem:msf_original} is proven using the following two lemmas, which are deterministic results on the MSF in the context of bins of items.

\begin{lemma}
\label{lem:msf_bin_sum}
 The MSF of a partial order is the sum of the MSF of each bin, \emph{i.e.}, for $\binseq$ a bin sequence, 
$$\msffunc(\permutset_\binseq)
=\sum_{\idxbin=1}^{|\binseq|} \msffunc(\permutset_{\bin_\idxbin})$$
\end{lemma}

\begin{proof}
Let us define the MSF of a set with respect to an ordering $\sigma^* \in \permutset$:

$$
\msffunc(\permutset, \sigma^*) 
    \triangleq \max_{\sigma \in \permutset} \sffunc(\sigma, \sigma^*)
$$

For all $\sigma \in \permutset_\binseq$, let $\sigma_{|\bin}$ be the order induced by $\sigma$ on bin $\bin$.
Let us prove the following preliminary result:

\begin{equation}
\label{eq:msfsum}
\forall \sigma^* \in \permutset_\binseq, 
 \msffunc(\permutset_{\binseq}, \sigma^*)
=\sum_{\idxbin=1}^{|\binseq|} \msffunc(\permutset_{\bin_\idxbin}, \sigma^*_{\mid \bin_\idxbin}) 
\end{equation}

    Let $\binseq = (\bin_1, \ldots, \bin_{|\binseq|})$ be the bin sequence. Let $\sigma^* \in \permutset_\binseq$. 
    Let $\sigma^- = \argmax_{\sigma \in \permutset_\binseq} SF(\sigma, \sigma^*)$ be the order for which the MSF is reached. 
    Then $\sigma^-$ is the worst possible ordering that respects this bin sequence. 
    We have:
    \begin{align*}
        \msffunc(\permutset_\binseq, \sigma^*) 
        &= SF(\sigma^-, \sigma^*)  \\
        &= \sum_{\idxitem \in [\nbi]} |\sigma^*(\idxitem) - \sigma^-(\idxitem)| \\
        &= \sum_{ \bin \in \binseq} \sum_{\idxitem \in \bin} |\sigma^*(\idxitem) - \sigma^-(\idxitem)|
    \end{align*}

    For any $\sigma \in \permutset_\binseq$, 
    $ \forall \idxbin \in [|\binseq|], \forall \idxitem \in \bin_\idxbin$, we have:
$$ \sigma_{|\bin_k}(\idxitem) = \sigma(\idxitem)                        - g(\idxbin(\idxitem))
$$
where, $g(\idxbin(\idxitem)) = \sum_{l<\idxbin(\idxitem)} |\bin_l|$ is the number of items in all bins before the one containing item $\idxitem$. \\
    
    This gives us :
    $$ \forall \bin \in \binseq, \forall \idxitem\in [\nbu
    i], 
     |\sigma_{|\bin}^*(\idxitem) - \sigma_{|\bin}^-(\idxitem)|
    = |(\sigma^*(\idxitem) - g(\idxbin(\idxitem)))   - (\sigma^-(\idxitem) - g(\idxbin(\idxitem))) | 
    = |\sigma^*(\idxitem) - \sigma^-(\idxitem)| 
    $$

    so $\forall \bin \in \binseq$,
    
    \begin{align*}
    \sum_{\idxitem \in \bin} |\sigma^*(\idxitem) - \sigma^-(\idxitem)|
    &= \sum_{\idxitem \in \bin} |\sigma_{|\bin}^*(\idxitem) - \sigma_{|\bin}^-(\idxitem)|\\
    &= \msffunc(\permutset_{\bin}, \sigma^*_{\mid \bin})
    \end{align*}

    So $\msffunc(\permutset_\binseq, \sigma^*) = \sum_{ \bin \in \binseq} \msffunc(\permutset_{\bin}, \sigma^*_{\mid \bin}) $,
    thus proving equation \eqref{eq:msfsum}.

 Using this results, we have:
\begin{align}
    \msffunc(\permutset_\binseq) 
    & \triangleq \max_{\sigma, \sigma^* \in \permutset_\binseq} \sffunc(\sigma, \sigma^*) \\
    &= \max_{\sigma^* \in \permutset_\binseq} \max_{\sigma \in \permutset_\binseq} \sffunc(\sigma, \sigma*) \\
    &= \max_{\sigma^* \in \permutset_\binseq} \msffunc(\permutset_\binseq, \sigma^*) \\
    & = \max_{\sigma^* \in \permutset_\binseq} \sum_{\idxbin=1}^{|\binseq|} \msffunc(\permutset_{\bin_\idxbin}, \sigma^*_{\mid \bin_\idxbin})  \\
    \label{lin:maxsum}
    & =  \sum_{\idxbin=1}^{|\binseq|} \max_{\sigma^* \in \permutset_{\bin_\idxbin}} \msffunc(\permutset_{\bin_\idxbin}, \sigma^*)  \\
    & =  \sum_{\idxbin=1}^{|\binseq|} \max_{\sigma^* \in \permutset_{\bin_\idxbin}} \max_{\sigma \in \permutset_{\bin_\idxbin}} \sffunc(\sigma, \sigma^*)  \\
    & =  \sum_{\idxbin=1}^{|\binseq|} \max_{\sigma, \sigma^* \in \permutset_{\binseq_\idxbin}}  \sffunc(\sigma, \sigma^*)  \\
    & =  \sum_{\idxbin=1}^{|\binseq|} \msffunc(\permutset_{\bin_\idxbin})
\end{align}

where the equality of line \eqref{lin:maxsum} is true because any order $\sigma$ on the full set can be decomposed in a sequence of one order on each bin.
\end{proof}

\begin{lemma} 
\label{lem:msf_bin}
     Let $\bin$ be a bin (\emph{i.e.} a finite set).\\
     Then
     $\msffunc(\permutset_{\bin}) $
     $= \frac{|\bin|^2}{2}$ if $|\bin|$ is even and 
     $\msffunc(\permutset_{\bin}) $
     $= \frac{|\bin|^2 -1}{2}$ if $|\bin|$ is odd.
\end{lemma}

 \begin{proof}
 For  $b \in \N$, we write $MSF_b$ the value of $\msffunc(\permutset_\bin)$ for $b=|\bin|$. \\
 Let us first consider the case where $\bin$ is even.\\
 The value of $MSF_b$ is equal to the Spearman footrule between two opposite orderings of $\bin$.
 The displacement between the first and the last element is $b-1$.
 The displacement between the second and the second to last elements  is $b-3$.
 This goes on until the displacement of $1$ on the two central elements.
 Consequently:

$$
MSF_b 
= (b-1) + (b-3) + \ldots + 3 + 1 + 1 + 3 + \ldots + (b-3) + (b-1)
= \frac{b^2}{2}
$$

 Using this, we can prove the result for $b$ odd. If the number of elements is even, and they are in reverse order (such that the SF is maximal), then adding an new element in the middle creates a new reverse ordering, and it adds one to the SF of each of the previous elements (the new element has an SF of 0 because it is in the middle). 
     So, 
     $\forall k \in \N$
     
\begin{align*}
MSF_{2k+1} 
= MSF_{2k} + 2k 
= \frac{4k^2 + 4k}{2} 
= \frac{(2k+1)^2 -1}{2} 
\end{align*}
\end{proof}

\subsubsection{Computing $\E[B^2]$}
\label{app:sec:b_2}

In this section, we prove Lemma \ref{lem:exp_square_bin} and its corollary.

\begin{replemma}{lem:exp_square_bin}
Let $\beta \in (0.5, 1)$.
Then

$$
\E[B^2]
=  \frac{\nbi}{\nbu+1} +  2  \frac{\nbi^2-\nbi}{(\nbu+1)^2} \E [  f_X(Y)^2]  +  O\left(\frac{\nbi^2}{\nbu^{2 \beta + 1}}  \right) 
$$

\end{replemma}

\begin{proof}

We use Lemmas \ref{lem:expectation_pk2}, \ref{lem:b_to_p} and \ref{lem:d2_cond_y}, all proven Appendix \ref{app:sec:properties}.
We have

\begin{align}
\sum_{k=0}^\nbu \E[B_k^2] 
&=  \nbi +  (\nbi^2 - \nbi) \sum_{k=0}^\nbu \E[P_k^2]  \\
\label{lin:pk2}
&=  \nbi +  (\nbi^2 - \nbi) \left( \sum_{k=0}^\nbu \E[D_k^2  f_X(Y_k)^2] + O\left(\frac{1}{\nbu^{2 \beta}} \right) \right)\\
\label{lin:d2f2}
&=  \nbi +  (\nbi^2 - \nbi) \left( (\nbu+1) \E[D^2  f_X(Y)^2] + O\left(\frac{1}{\nbu^{2 \beta}} \right) \right)
\end{align}

where we used Lemma \ref{lem:b_to_p} for the first equality and Lemma \ref{lem:expectation_pk2} for the second.

Furthermore, we have
\begin{align}
& \E[D^2  f_X(Y)^2] \\
& = \E[  f_X(Y)^2 \E[D^2| Y]]  \\
&  = \E[  f_X(Y)^2 \left( \E[D^2| Y, Y\neq0] \prob(Y \neq 0) + \E[D^2| Y, Y = 0 ] \prob(Y = 0) ] \right) ]\\
\label{lin:d2}
&  =  \E\left[  f_X(Y)^2 \left(   2 \frac{ 1 - Y^{\nbu+1} - (\nbu+1) Y^{\nbu}(1-Y)}{(\nbu+1)\nbu} \frac{\nbu}{\nbu+1}  + \frac{2}{(\nbu+1)(\nbu+2)}  \frac{1}{\nbu+1} \right) \right]  \\
\label{lin:devy}
& = \frac{2}{(\nbu+1)^2}   \E \left[    f_X(Y)^2 \left( 1 - Y^{\nbu+1} - (\nbu+1) Y^{\nbu}(1-Y) \right)\right]   + \frac{2}{(\nbu+1)^2(\nbu+2)}
\end{align}
 
where \eqref{lin:d2} comes from Lemma \ref{lem:d2_cond_y}.\\

In addition, using $c_X \triangleq \max f_X$, we have

\begin{align*}
\E [  f_X(Y)^2]
& \geq \E \left[  f_X(Y)^2  \left( 1 - Y^\nbu - \nbu Y^{\nbu-1}(1-Y) \right)\right]  \\
&=\E [  f_X(Y)^2] 
-  \E[f_X(Y)^2 Y^\nbu] 
- \nbu \E [f_X(Y)^2  Y^{\nbu-1}(1-Y) ] \\
&\geq \E [  f_X(Y)^2] 
-  c_X^2 \E[ Y^\nbu] 
-  c_X^2 \nbu \E [Y^{\nbu-1}(1-Y) ] \\
&= \E [  f_X(Y)^2] 
-  c_X^2 \frac{1}{\nbu+1}
-  c_X^2 \nbu \frac{1}{\nbu(\nbu+1)} \\
&= \E [  f_X(Y)^2]  -  \frac{2 c_X^2}{\nbu+1}
\end{align*}

so, combining with \eqref{lin:devy},

$$
\frac{2}{(\nbu+1)^2} \left( \E [  f_X(Y)^2]  -  \frac{2 c_X^2}{\nbu+2} \right) + \frac{2}{(\nbu+1)^2(\nbu+2)}
\leq \E[D^2  f_X(Y)^2] \\
\leq \frac{2}{(\nbu+1)^2} \E [  f_X(Y)^2]  +  \frac{2}{(\nbu+1)^2(\nbu+2)}
$$

which gives

$$
\E[D^2  f_X(Y)^2] 
=  \frac{2}{(\nbu+1)^2} \E [  f_X(Y)^2]  + O\left(\frac{1}{\nbu^3}\right)
$$

so using \eqref{lin:d2f2}, we obtain

\begin{align*}
\sum_{k=0}^\nbu \E[B_k^2]  
&=  \nbi +  (\nbi^2 - \nbi) \left( (\nbu+1) \E[D^2  f_X(Y)^2] + O\left(\frac{1}{\nbu^{2 \beta}} \right) \right) \\
&=  \nbi +  (\nbi^2 - \nbi) \left( (\nbu+1) \left(\frac{2}{(\nbu+1)^2} \E [  f_X(Y)^2]  + O\left(\frac{1}{\nbu^3}\right) \right) + O\left(\frac{1}{\nbu^{2 \beta}} \right) \right) \\
&=  \nbi +  (\nbi^2 - \nbi) \left( \frac{2}{\nbu+1} \E [  f_X(Y)^2]  + O\left(\frac{1}{(\nbu+1)^2} \right) + O\left(\frac{1}{\nbu^{2 \beta}} \right) \right) \\
&=  \nbi +  (\nbi^2 - \nbi) \left( \frac{2}{\nbu+1} \E [  f_X(Y)^2]  +  O\left(\frac{1}{\nbu^{2 \beta}} \right) \right) \\
&=  \nbi +  2  \frac{\nbi^2 - \nbi}{\nbu+1} \E [  f_X(Y)^2]  +  O\left(\frac{\nbi^2}{\nbu^{2 \beta}}  \right) \\
\end{align*}

and finally,

\begin{align*}
\E[B^2] 
&= \frac{1}{\nbu+1} \sum_{k=0}^\nbu \E[B_k^2] \\
&=  \frac{\nbi}{\nbu+1} +  2  \frac{\nbi^2-\nbi}{(\nbu+1)^2} \E [  f_X(Y)^2]  +  O\left(\frac{\nbi^2}{\nbu^{2 \beta + 1}}  \right) \\
\end{align*}

\end{proof}

\begin{corollary}
\label{cor:exp_square_bin}

Let $m=\Omega(n)$, then

$$
\E[B^2] 
= \frac{\nbi}{\nbu} +  2  \frac{\nbi^2}{\nbu^2} \E [  f_X(Y)^2]  +  o\left(\frac{\nbi^2}{\nbu^2} \right)
$$
\end{corollary}

\begin{proof}
    From the previous lemma:

\begin{align*}
\E[B^2] 
&=  \frac{\nbi}{\nbu+1} +  2  \frac{\nbi^2-\nbi}{(\nbu+1)^2} \E [  f_X(Y)^2]  +  O\left(\frac{\nbi^2}{\nbu^{2 \beta + 1}}  \right) \\
&=  \frac{\nbi}{\nbu} +  2  \frac{\nbi^2}{\nbu^2} \E [  f_X(Y)^2]  +  o\left(\frac{\nbi^2}{\nbu^{2}}  \right) 
\end{align*}

\end{proof}

\subsubsection{Computing $\prob(B \isodd)$}
\label{app:sec:b_odd}

In this section, we prove Lemma \ref{lem:proba_b_odd}.
For this, we need Lemma \ref{lem:huge}, which itself uses the result of Lemma \ref{lem:hell}.
Both are proven below in this section.

\begin{replemma}{lem:proba_b_odd}
Let $\beta \in (0.5, 1)$, $\nbu = \omega(\nbi)$. 
Then, for $\nbi$ going to infinity,
$$
\prob(B \isodd) 
= \frac{\nbi}{\nbu} - 2\frac{\nbi^2}{\nbu^2}\E [  f_X(Y)^2]  
+ O \left(\frac{n}{m^{2\beta}} \right)  + o \left(\frac{n^2}{m^2}  \right) \\
$$
\end{replemma}

\begin{proof}

$B$ conditional on both $Y$ and $D$ follows a Binomial distribution of parameters $(n, P \triangleq p(Y,D))$, where $p(Y,D)$ is the probability of an item being in bin $K$, as defined in Definition \ref{def:pk}.

The probability that a binomial random variable of parameters $(n, p)$ is odd is  
$ \frac{1}{2} (1 - ( 1-2p)^n)$. \\

So
\begin{align*}
\prob(B \isodd | Y, D)
&= \frac{1}{2} (1 - ( 1-2p(Y, D))^\nbi) \\
&= \frac{1}{2} (1 - ( 1-2P)^\nbi) 
\end{align*}

By linearity of the expectation
$$
\prob(B \isodd) 
= \E[ \prob(B \isodd | Y, D) ]
= \frac{1}{2} (1 - \E[( 1-2P)^\nbi]) 
$$

So, using Lemma \ref{lem:huge} (below),

\begin{align*}
\prob(B \isodd) 
&= \frac{1}{2}( 1- (\E[(1-2P)^n])) \\
& = \frac{1}{2} \left( 1- \left(
  1  - 2 \frac{n}{m} + 4 \frac{n^2}{m^2} \E[f_X(Y)^2] 
+ O \left(\frac{n}{m^{2\beta}} \right)  + o \left(\frac{n^2}{m^2} \right) \right) \right) \\
& = \frac{n}{m} - 2 \frac{n^2}{m^2} \E[f_X(Y)^2] 
+ O \left(\frac{n}{m^{2\beta}} \right)  + o \left(\frac{n^2}{m^2}  \right) 
\end{align*}

\end{proof}

We need the following result to prove Lemma \ref{lem:proba_b_odd}:

\begin{lemma}
\label{lem:huge}
Let $\beta \in (0.5, 1)$.
Let $\nbu = \omega(\nbi)$ as $\nbi$ goes to infinity. Then we have 
$$
\E[(1-2P)^n]
=  1  - 2 \frac{n}{m} + 4 \frac{n^2}{m^2} \E[f_X(Y)^2] 
+ O \left(\frac{n}{m^{2\beta}} \right)  + o \left(\frac{n^2}{m^2} \right) 
$$
\end{lemma}

\begin{proof}

Let $\beta \in (0.5, 1)$. 
Let $c_X \triangleq \max f_X$ and, for all $m \in \N$, let $c_m(\beta) \triangleq \pm \frac{2c}{m^{2\beta}}$, where $c$ is the Lipschitz constant of $f_X$.
To avoid clutter, we use the notation $c_m \triangleq  c_m(\beta)$.
We recall the definition of event $\event$:

$$
\forall \beta \in (0,1),  \nbu \in \N, \quad
\mathcal{E}(\beta, \nbu) 
\triangleq \left( \bigcap_{k=0}^\nbu \left( D_k \leq  \frac{1}{\nbu^\beta} \right) \right) 
$$

We will use the following structure to estimate  $\E[(1-2P)^n] $:
\begin{align*}
\E[(1-2P)^n] 
&\simeq \E[(1-2P)^n | \event] \\
&\simeq \E[(1-2\min(Df_X(Y)+ \frac{c_m}{2},1) )^n | \event] \\
&\simeq \E[(1-2\min(Df_X(Y)+ \frac{c_m}{2},1) )^n ] \\
&\simeq \E[(1-2Df_X(Y) + c_m)^n ]   \\
&\simeq 1  - 2 \frac{n}{m} + 4 \frac{n^2}{m^2} \E[f_X(Y)^2]
\end{align*}

More precisely, we will prove the corresponding approximations:
\begin{enumerate}
\item $
|\E[(1-2P)^n] -    \E[(1-2P)^n | \event(\beta, m)]| \leq 2(m+1) e^{-m^{1-\beta}}
$

\item $\E[(1-2P)^n | \event(\beta, m)]|$
is bounded by 
$\E[(1-2\min(Df_X(Y)\pm \frac{c_m}{2}, 1) )^n | \event(\beta ,m)]$

\item $
|\E[(1-2\min(Df_X(Y) + \frac{c_m}{2}, 1 ))^n | \event(\beta ,m)]
- \E[(1-2\min(Df_X(Y)+ \frac{c_m}{2}, 1 ))^n ] | 
\leq 2(m+1) e^{-m^{1-\beta}}
$

\item $
|\E[(1-2\min(Df_X(Y)+ \frac{c_m}{2}, 1))^n ]  
- \E[(1-2Df_X(Y) + c_m)^n ] | 
\leq   e^{-m/c_X + o(m)}
$

\item $
\E[(1-2Df_X(Y) + c_m)^n ]   
=  1  - 2 \frac{n}{m} + 4 \frac{n^2}{m^2} \E[f_X(Y)^2] 
+ O \left(\frac{n}{m^{2\beta}} \right)  + o \left(\frac{n^2}{m^2} \right) 
$
\end{enumerate}


Combining these 5 points directly yields the final result:
$$
\E[(1-2P)^n]
=  1  - 2 \frac{n}{m} + 4 \frac{n^2}{m^2} \E[f_X(Y)^2] 
+ O \left(\frac{n}{m^{2\beta}} \right)  + o \left(\frac{n^2}{m^2} \right) 
$$

\paragraph{Proof of 5.} 5. is immediate using Lemma \ref{lem:hell} (below), by setting $v=2c$, where $c$ is the Lipschitz constant of $f_X$.

\paragraph{Proof of 1. and 3.}

We use Lemma \ref{lem:conditional_expectation} for both of them.

Let us set 
$W= (1-2P)^n$
and
$Z= (1-2\min(Df_X(Y)+ c_m, 1 ))^n$.

We have $|W|\leq 1$ and $|Z|\leq 1$.
So, using Lemma \ref{lem:conditional_expectation},
$|\E[W] - \E[W|\event] | \leq 2 (1-\prob(\event))
= 2(\nbu+1) e^{-\nbu^{\beta-1}}$
and
$|\E[Z] - \E[Z|\event] | \leq 2 (1-\prob(\event))
= 2(\nbu+1) e^{-\nbu^{\beta-1}}$.

\paragraph{Proof of 4.}

For all $m,n\in \N$,

\begin{align*}
& \left|\E[(1-2\min(Df_X(Y)+ \frac{c_m}{2}, 1 ))^n ]   - \E[(1-2(Df_X(Y) + \frac{c_m}{2}))^n ] \right| \\
&= \left| \E \left[\mathbb{1}(Df_X(Y)+ \frac{c_m}{2}\geq 1) \left(\E[(1-2\min(Df_X(Y)+ \frac{c_m}{2}, 1 ))^n ]   - \E[(1-2Df_X(Y) + c_m)^n ] \right) \right] \right| \\
&=\prob(Df_X(Y)+ \frac{c_m}{2}\geq 1) \left| \E[(1-2\min(Df_X(Y)+ \frac{c_m}{2}, 1 ))^n ]   - (1-2Df_X(Y) + c_m)^n ) |Df_X(Y)+ \frac{c_m}{2}\geq 1 ] \right| \\
&=\prob(Df_X(Y)+ \frac{c_m}{2}\geq 1)  \left| \E[(1-2)^n ]   - (1-2Df_X(Y) + c_m)^n ) |Df_X(Y)+ \frac{c_m}{2}\geq 1 ] \right| \\
&\leq \prob \left( D\geq \frac{1-\frac{c_m}{2}}{f_X(Y)} \right)  \E[(2+2 c_X+ c_m)^n |Df_X(Y)+ \frac{c_m}{2}\geq 1 ] \\
& \leq  \prob(D\geq \frac{1-\frac{c_m}{2}}{c_X})  (2+2 c_X+ c_m)^n  \\
& \leq  e^{-\frac{1-\frac{c_m}{2}}{c_X}m}  e^{O(n)} \\
&\leq   e^{-m/c_X + o(m)}
\end{align*}

where we used Lemma \ref{lem:bound_beta} with $\beta=0$ and $a=\frac{1-\frac{c_m}{2}}{c_X}$ ($D$ is a $Beta(1, m)$ because it is a mixture of $m+1$ different $Beta(1, m)$).

\paragraph{Proof of 2.}
$f$ is upper bounded by $c_X$ and $\event$ implies $D \leq \frac{c}{m^\beta}$ (Definition \ref{def:event}).\\ 
So, conditional on $\event$, for $m$ large enough, we have
$|2Df_X(Y)| +  |\frac{2c}{m^{2\beta}}| \leq 1$, which implies \\
$(1-2Df_X(Y) \pm \frac{2c}{m^{2\beta}}) \geq 0$. 

We also have $0 \leq P \leq 1$, so $\mathcal{E}$  implies $|P- \min(Df_X(Y),1) | \leq |P- Df_X(Y)| \leq \frac{c}{m^{2\beta}}$ (Lemma \ref{lem:proba_asymp}).
 
So, for $m$ large enough, for all $n \in \N$, conditioned on $\event(\beta, m)$,

\begin{alignat*}{4}
    0
   & \leq Df_X(Y) -  \frac{c}{m^{2\beta}}
   && \leq P
   &&& \leq Df_X(Y) +  \frac{c}{m^{2\beta}} \\
   0
   & \leq \min(Df_X(Y) -  \frac{c}{m^{2\beta}}, 1)
   && \leq P
   &&& \leq \min(Df_X(Y) +  \frac{c}{m^{2\beta}}, 1) 
\end{alignat*}

So, for $m$ large enough, for all $n \in \N$,

$$
 \E\left[(1-2\min(Df_X(Y) - \frac{c}{m^{2\beta}}, 1) )^n | \event \right]
\geq \E\left[(1 - 2P)^n | \event \right]
\geq \E\left[(1-2\min(Df_X(Y)+ \frac{c}{m^{2\beta}}, 1))^n| \event \right] 
$$

We conclude the proof of \textbf{2.} by recalling that we defined $c_m \triangleq \pm \frac{2c}{m^{2\beta}}$.

\end{proof}

\begin{lemma}
\label{lem:hell}
Let  $c_m \triangleq \frac{v}{m^{2\beta}}$, where $\beta > 0.5$ and $v>0$ are constants.
Assume that 
$m = \omega(n)$.

Then we have

$$
\E[(1-2Df_X(Y) + c_m)^n]
=  1  - 2 \frac{n}{m} + 4 \frac{n^2}{m^2} \E[f_X(Y)^2] 
+ O\left(\frac{n}{m^{2\beta}}\right) + o \left(\frac{n^2}{m^2} \right) 
$$

\end{lemma}

\begin{proof}

We start by proving the result for $c_m=0$, and we extend the analysis to $c_m \triangleq \frac{v}{m^{2\beta}}$ in a second part.
We have

\begin{equation}
\label{eq:hell_init}
\E[(1-2Df_X(Y))^n]  
= 1- 2n \E[f_X(Y)D] + 2n(n-1) \E[(f_X(Y)D)^2]   + H
\end{equation}

where $H \triangleq \sum_{k=3}^n \binom{n}{k} (-2)^k \E[(f_X(Y)D)^k]$ is the error term.

We will compute the two expectations, and upper bound $|H|$.

Lemmas \ref{lem:d1_cond_y} and \ref{lem:d2_cond_y} respectively give us the first and second moment of $D$ conditioned on $Y$:

$$\E[D | Y=y>0] = \frac{1}{m} - \frac{y^m}{m}, \quad
\E[D | Y=0] = \frac{1}{m+1}
$$
$$\E[D^2 | Y=y>0] = \frac{2}{m(m+1)} (1 - y^{m+1} - (m+1) (1-y)y^m), 
\quad
\E[D^2 | Y=0] = \frac{2}{(m+1)(m+2)}
$$

So, by the law of iterated expectations, we have the first expectation of \eqref{eq:hell_init}:

\begin{align*}
\E[f_X(Y) D]
&= \prob(Y>0) \E[f_X(Y) D | Y>0] + \prob(Y=0) \E[f_X(Y) D | Y=0] \\
&= \frac{m}{m+1} \E\left[ f_X(Y) \left( \frac{1}{m} - \frac{Y^m}{m} \right)  \right] + \frac{1}{m+1} O\left(\frac{1}{m+1}\right) \\
&= \frac{m}{m+1} \left(\frac{1}{m}\E[ f_X(Y)] -  \frac{1}{m} \E[f_X(Y) Y^m] \right) + O\left(\frac{1}{m^2}\right) \\
&= \frac{1}{m} + O\left(\frac{1}{m^2}\right)
\end{align*}

and the second expectation of \eqref{eq:hell_init}:
\begin{align*}
\E[(f_X(Y) D)^2] 
&= \prob(Y>0) \E[ (f_X(Y) D)^2 | Y>0] + \prob(Y=0) \E[(f_X(Y) D)^2 | Y=0] \\
&=  \frac{m}{m+1} \E\left[ f_X(Y)^2 \frac{2}{m(m+1)} (1 - Y^{m+1} - (m+1) (1-Y)Y^m) \right] 
+  \frac{1}{m+1} O \left( \frac{2}{(m+1)(m+2)} \right) \\
&= \frac{2}{(m+1)^2} \E [f_X(Y)^2] + \frac{2}{(m+1)^2} \E\left[ f_X(Y)^2 (Y^{m+1} - (m+1) (1-Y)Y^m)) \right] 
+ O\left(\frac{1}{m^3}\right)
\\
&= \frac{2}{m^2}\E [f_X(Y)^2] + O\left(\frac{1}{m^3}\right)
\end{align*}

Using these two results, we develop \eqref{eq:hell_init} to obtain

\begin{align}
\E[(1-2Df(Y))^n] 
&= 1- 2n \left(\frac{1}{m} + O\left(\frac{1}{m^2} \right)  \right)  
+ 2n(n-1) \left(\frac{2}{m^2}\E[f_X(Y)^2] + O\left(\frac{1}{m^3}\right) \right)  
+ H \\
&=  1  - 2 \frac{n}{m} + 4 \frac{n^2}{m^2} \E[f_X(Y)^2] 
+ o \left(\frac{n^2}{m^2} \right) 
+ H
\label{lin:H_remains}
\end{align}

Regarding the error term $H$, we have

$$
|H| 
\triangleq \left| \sum_{k=3}^n \binom{n}{k} (-2)^k \E[(f_X(Y)D)^k] \right| 
\leq  \sum_{k=3}^n \binom{n}{k}     (2c_X)^k \E[D^k]  
$$

where $c_X$ is the upper bound of $f_X$.

Because $D$ is a $Beta(1,m)$, we have $\E[D^k] = \frac{k! m!}{(k+m)!}
= \frac{1}{\binom{k+m}{m}}
$

So 
\begin{equation}
|H| 
\leq \sum_{k=3}^n \frac{\binom{n}{k}}{\binom{k+m}{k}}   (2c_X)^k
\end{equation}

Lemma \ref{lem:binomial_bound} gives us 
$ \frac{\binom{n}{k}}{\binom{k+m}{k}}
\leq  \frac{e^2}{2\pi} \left( \frac{n}{m} \right)^k
$, so we have

\begin{align*}
|H|
&\leq \frac{e^2}{2 \pi} \sum_{k=3}^n \left( \frac{n}{m} \right)^k   (2c_X)^k \\
&\leq \frac{e^2}{2\pi} \sum_{k=3}^\infty \left(\frac{n}{m} 2c_X \right)^k \\
&= \frac{e^2}{2\pi} \frac{\left(\frac{n}{m} 2c_X \right)^3  }{1-\frac{n}{m} 2c_X}  \\
&=  O((n/m)^3) \\
&=  o((n/m)^2)
\end{align*}
where the last two equalities come from our assumption $m = \omega(n)$.

We replace $H$ in \eqref{lin:H_remains} to obtain

\begin{equation}
\label{eq:without_cm}
\E[(1-2Df(Y))^n]
=  1  - 2 \frac{n}{m} + 4 \frac{n^2}{m^2} \E[f_X(Y)^2] 
+ o \left(\frac{n^2}{m^2} \right) 
\end{equation}

We now compute the error created when putting the margin term $c_m$. \\
We will show that
$
\E[(1-2Df(Y) + c_m)^n] 
- \E[(1-2Df(Y))^n]  
= O \left(\frac{n}{m^{2\beta}} \right)  + o \left(\frac{n^2}{m^2} \right) 
$
.

We have

\begin{align*}
&\E[(1-2Df(Y) + c_m)^n] 
- \E[(1-2Df(Y))^n] \\
&= n c_m + \frac{n(n-1)}{2} ( c_m \E[2Df(Y)] + c_m^2)
+ H' - H \\
&= n c_m + \frac{n(n-1)}{2} \left(2 c_m \left(\frac{1}{m} + O\left(\frac{1}{m^2}\right) \right) + c_m^2 \right)
+ H' - H \\
&= O\left(\frac{n}{m^{2\beta}}\right) + \frac{n(n-1)}{2} \left(O\left(\frac{1}{m^{2\beta+1}}\right) + O\left(\frac{1}{m^{2( 2\beta)}}\right) \right)
+ H' - H \\
&= O\left(\frac{n}{m^{2\beta}}\right) + H' - H
\end{align*}

where $H'= \sum_{k=3}^n \binom{n}{k} \E[ (-2 f_X(y)D  + c_m)^k ] $.\\

We know that $H = o \left( \frac{n^2}{m^2} \right)$,  so we have

\begin{equation}
\label{eq:exp_diff}
\E[(1-2Df(Y) + c_m)^n]  - \E[(1-2Df(Y))^n] 
= H' + O\left(\frac{n}{m^{2\beta}}\right) + o\left(\frac{n^2}{m^2}\right) 
\end{equation}

In addition, we have

\begin{align*}
|H'|
& \leq  \sum_{k=3}^n \binom{n}{k} \E[ (2 f_X(Y)D  + c_m)^k ] \\
& \leq \sum_{k=3}^n \binom{n}{k}   \sum_{j=0}^{k} \binom{k}{j} (2 c_X)^j \E[D^j] c_m^{k-j}  \\
& = \sum_{k=3}^n \binom{n}{k}  \left( (2c_X)^k \E[D^k] + \sum_{j=0}^{k-1} \binom{k}{j} (2 c_X)^j \E[D^j] c_m^{k-j} \right) \\
& = o \left( \frac{n^2}{m^2} \right) + \sum_{k=3}^n \binom{n}{k}   \sum_{j=0}^{k-1} \binom{k}{j} (2 c_X)^j \E[D^j] c_m^{k-j}  \\
& = o \left( \frac{n^2}{m^2} \right) + \sum_{k=3}^n \binom{n}{k} \sum_{j=0}^{k-1} \binom{k}{j} (2 c_X)^j \frac{1}{\binom{j+m}{m}} \left(\frac{2c}{m^{2\beta}} \right)^{k-j} \\
& \leq o \left( \frac{n^2}{m^2} \right) + \sum_{k=3}^n \binom{n}{k} \sum_{j=0}^{k-1}   \frac{\binom{k}{j}}{\binom{j+m}{m}} \left( \frac{2c}{m^{2\beta-1}} \right)^{k-j}  \frac{1}{m^{k-j}} (2 c_X)^j\\
\end{align*}

Using $\frac{\binom{k}{j}}{\binom{j+m}{m}} \leq \frac{e^2}{\pi} \left( \frac{k}{m} \right)^j$ from Lemma \ref{lem:binomial_bound} and for $m$ large enough (such that $m^{2\beta-1} \geq 2c$), we have

\begin{align*}
|H'|
& \leq o \left( \frac{n^2}{m^2} \right) + \sum_{k=3}^n \binom{n}{k} \sum_{j=0}^{k-1} \left( \frac{k}{m} \right)^j   \frac{1}{m^{k-j}} (2 c_X)^j\\
& \leq o \left( \frac{n^2}{m^2} \right) + \sum_{k=3}^n \binom{n}{k}  \frac{1}{m^k} \sum_{j=0}^{k-1}  k^j (2 c_X)^j  \\
& \leq o \left( \frac{n^2}{m^2} \right) + \sum_{k=3}^n \left( \frac{ne}{k} \right)^k  \frac{1}{m^k} \sum_{j=0}^{k-1}  k^j (2c_X)^j  \\
& \leq o \left( \frac{n^2}{m^2} \right) + \sum_{k=3}^n \left( \frac{ne}{m} \right)^k  \frac{1}{k^k} \frac{k^k (2c_X)^k - 1}{k(2c_X)-1} \\
& \leq o \left( \frac{n^2}{m^2} \right) + \sum_{k=3}^n \left( \frac{ne(2c_X)}{m} \right)^k  \\
& = o \left( \frac{n^2}{m^2} \right)
\end{align*}

where we used the general inequality $\binom{n}{k} \leq \left( \frac{ne}{k} \right)^k$.

Using this in \eqref{eq:exp_diff}, we get that 
$
\E[(1-2Df(Y) + c_m)^n]  - \E[(1-2Df(Y))^n] 
= O(\frac{n}{m^{2\beta}}) + o(\frac{n^2}{m^2}) 
$, so from \eqref{eq:without_cm}, we obtain our final result:

$$
\E[(1-2Df(Y) + c_m)^n]
=  1  - 2 \frac{n}{m} + 4 \frac{n^2}{m^2} \E[f_X(Y)^2] 
+ O\left(\frac{n}{m^{2\beta}}\right) + o \left(\frac{n^2}{m^2} \right) 
$$
\end{proof}

\subsection{Proof of Theorem \ref{thm:msf_high}}
\label{app:sec:msf_high}

Lemma \ref{lem:msf_original} gives us 
$\E[F] = \frac{1}{2} (\nbu +1) (\E[B^2] - \prob(B \isodd))$.
In this expression, the last term is at most one, so it can be ignored if $\E[B^2]$ goes to infinity. 
This is the case for Theorem \ref{thm:msf_high}.
Theorem \ref{thm:msf_low} treats the case of a constant MSF, for which we need this second term.

\begin{reptheorem}{thm:msf_high}[For $f_Y=1$]
Assume that there exists $r \in \R^+ \emph{ s.t. } \nbu \sim r \nbi$ as $\nbi$ goes to infinity, then
$$
\left| \E[F] - \nbi \left(\frac{1}{2} +  \frac{1}{r}  \left[f_X(Y)^2\right] \right) \right|
\leq \frac{r}{2} \nbi + o(\nbi) 
$$
\end{reptheorem}

\begin{proof}

Lemma \ref{lem:msf_original} gives 
$ \E[F] = \frac{1}{2} (\nbu+1) ( \E[B^2] - \prob(B \isodd)) $
and we have \\
$ 1 \geq \prob(B \isodd)) \geq 0 $,
so
$ \frac{\nbu+1}{2} (\E[B^2] - 1 ) \leq \E[F]  \leq \frac{\nbu+1}{2} \E[B^2] $
and

$$
- \frac{\nbu+1}{2} 
\leq \E[F] - \frac{\nbu+1}{2} \E[B^2] 
\leq 0 \\
$$

From Lemma \ref{lem:exp_square_bin}, we have
$$
\forall \beta \in (0.5, 1), \quad \E[B^2] 
= \frac{\nbi}{\nbu+1}  + 2 \frac{\nbi^2-\nbi}{(\nbu+1)^2} \E[f_X(Y)^2] + O\left(\frac{\nbi^2}{\nbu^{2\beta+1}}\right)
$$

So, put back in the previous equation,
\begin{align}
\label{lin:msf_high}
- \frac{\nbu+1}{2}
&\leq \E[F] - \left( \frac{1}{2} \nbi  +  \frac{\nbi^2 - \nbi}{\nbu+1} \E[f_X(Y)^2] + O\left(\frac{\nbi^2}{\nbu^{2\beta}}\right) \right)
&\leq 0 
\end{align}

We have $\exists r \in \R^+ \emph{ s.t. }   \nbu \sim r \nbi$, so
$\nbu = r \nbi + o(\nbi)$. \\
Focusing on the middle term of equation \eqref{lin:msf_high},
\begin{align*}
&\frac{1}{2} \nbi  +  \frac{\nbi^2 - \nbi}{\nbu} \E[f_X(Y)^2] + O\left(\frac{\nbi^2}{\nbu^{2\beta}} \right) \\
& = \frac{1}{2} \nbi  +  \frac{\nbi^2 - \nbi}{(r \nbi + o(\nbi))} \E[f_X(Y)^2] + O\left(\frac{\nbi^2}{(r \nbi + o(\nbi))^{2\beta}} \right) \\
& = \frac{1}{2} \nbi  +  \frac{\nbi^2 - \nbi}{r\nbi} (1 + o(1)) \E[f_X(Y)^2] + O(n^{2-2\beta}) \\
& = \frac{1}{2} \nbi  +  \frac{\nbi - 1}{r}  \E[f_X(Y)^2] + O(n^{2(1-\beta)}) \\
& = \nbi \left(\frac{1}{2} +  \frac{1}{r}  \E[f_X(Y)^2] \right) + o(n)
\end{align*}

Putting the result back in equation \eqref{lin:msf_high}, we have \\
$
- \frac{r \nbi + o(\nbi)}{2}
\leq  \nbi \left(\frac{1}{2} +  \frac{1}{r}  \E[f_X(Y)^2] \right) + o(n) 
\leq 0 
$,
which implies
$$
\left| \E[F] - \nbi \left(\frac{1}{2} +  \frac{1}{r}  \E[f_X(Y)^2] \right) \right|
\leq \frac{r}{2} \nbi + o(\nbi) 
$$

\end{proof}

This version of Theorem \ref{thm:msf_high} depends on the \emph{rescaled} density $f_X$, \emph{i.e.} under the assumption that $f_Y=1$.
In order to generalize to a general $f_Y$, we use Lemma \ref{lem:rescale} (Appendix \ref{app:sec:divergence}). 
It states that if $(f_{X}, f_{Y})$ are the true densities of the scores and thresholds, and $(f_{X'}, 1)$ are their rescaled counterparts, then
$
\E[f_{X'}(Y')^2 ]
= \E\left[ \left( \frac{f_{X}(Y)}{f_{Y}(Y)} \right)^2  \right]
$.
This gives us the final version of the theorem:


\begin{reptheorem}{thm:msf_high}[For any $f_Y$]
Assume that there exists $r \in \R^+ \emph{ s.t. } \nbu \sim r \nbi$ as $\nbi$ goes to infinity, then
$$
\left| \E[F] - \nbi \left(\frac{1}{2} +  \frac{1}{r}  \dvg{X}{Y} \right) \right|
\leq \frac{r}{2} \nbi + o(\nbi) 
$$
\end{reptheorem}

\subsection{Proof of Theorem \ref{thm:msf_low}}
\label{app:sec:msf_low}

We provide here the proof of Theorem \ref{thm:msf_low}, which uses Lemmas \ref{lem:msf_original}, \ref{lem:exp_square_bin} (Corollary \ref{cor:exp_square_bin}) and \ref{lem:proba_b_odd}.

\begin{reptheorem}{thm:msf_low}[For $f_Y=1$]
Assume that there exists $ r \in \R^+, \gamma > 1 $ \emph{ s.t. } $\nbu \sim r \nbi^{\gamma}$, as $n$ goes to infinity, then
$$
\E[F]   
\sim \frac{2}{r} n^{2-\gamma} \E[f_X(Y)^2] 
$$
\end{reptheorem}

\begin{proof}

Let us set $m = r n^\gamma + o(n^\gamma)$, $\gamma > 1$ \\

Using Corollary \ref{cor:exp_square_bin} and Lemma \ref{lem:proba_b_odd}, for all $\beta \in (0.5, 1)$, we have

Corollary \ref{cor:exp_square_bin}:

$$
\E[B^2]
=  \frac{\nbi}{\nbu} +  2  \frac{\nbi^2}{\nbu^2} \E [  f_X(Y)^2]  +  o\left(\frac{\nbi^2}{\nbu^{2}}  \right) 
$$

Then Lemma \ref{lem:proba_b_odd}:

$$
\prob(B \isodd) 
= \frac{\nbi}{\nbu} - 2\frac{\nbi^2}{\nbu^2}\E [  f_X(Y)^2]  
+ O \left(\frac{n}{m^{2\beta}} \right)  + o \left(\frac{n^2}{m^2}  \right) \\
$$

So, using Lemma \ref{lem:msf_original}, we have
    
\begin{align}
\E[F] 
& = \frac{(\nbu+1)}{2} ( \E[B^2] - \prob(B\isodd)) \\
& = \frac{(\nbu+1)}{2} \left( \frac{\nbi}{\nbu} +  2  \frac{\nbi^2}{\nbu^2} \E [  f_X(Y)^2]    
- \frac{\nbi}{\nbu} 
+ 2\frac{\nbi^2}{\nbu^2} \E[f_X(Y)^2]  
+ O \left(\frac{n}{m^{2 \beta}} \right)  + o \left(\frac{n^2}{m^2}  \right)  \right ) \\
& =  2 \frac{\nbi^2}{\nbu} \E[f_X(Y)^2]  
+ O \left(\frac{n}{m^{2\beta-1}} \right)  + o \left(\frac{n^2}{m}  \right ) \\
\end{align}

By reparametrizing, for any $\alpha <1 $, we have
$$
\E[F]
=  2 \frac{\nbi^2}{\nbu} \E[f_X(Y)^2]  
+ O \left(\frac{n}{m^{\alpha}} \right)  + o \left(\frac{n^2}{m}  \right ) \\
$$

If it were true for $\alpha =1$, we would have 
$\E[F]
=  2 \frac{\nbi^2}{\nbu} \E[f_X(Y)^2]  
+  o \left(\frac{n^2}{m}  \right ) $ as long $m=\omega(n)$.

However, here we have to pick the value of $\alpha$ close enough to $1$, depending on the value of $\gamma$.

$ O\left( \frac{n}{m^\alpha} \right)
= O\left( \frac{n}{(rn^\gamma + o(n^\gamma))^\alpha} \right)
= O\left( n^{1-\alpha \gamma}\right)
$
and 
$
o \left(\frac{n^2}{m}  \right )
=o \left(n^{2-\gamma}  \right )
$

So by setting for instance $\alpha = 1 - \frac{1}{2\gamma}$, we have 
$ O\left( \frac{n}{m^\alpha} \right)
=o \left(\frac{n^2}{m}  \right )
$

and 

$$
\E[F]
=  2 \frac{\nbi^2}{\nbu} \E[f_X(Y)^2]   + o \left(\frac{n^2}{m}  \right ) 
$$

\emph{i.e.}

$$
\E[F]
\sim  2 \frac{\nbi^2}{\nbu} \E[f_X(Y)^2] 
$$

and

\begin{align}
\E[F]  
\sim \frac{2}{r} n^{2-\gamma} \E[f_X(Y)^2] 
\end{align}

\end{proof}

%
%
%

As we did for Theorem \ref{thm:msf_high}, we generalize to a general $f_Y$ using Lemma \ref{lem:rescale}, thus obtaining the final form of Theorem \ref{thm:msf_low}.

\begin{reptheorem}{thm:msf_low}[For general $f_Y$]
Assume that there exists $ r \in \R^+, \gamma > 1 $ \emph{ s.t. } $\nbu \sim r \nbi^{\gamma}$, as $n$ goes to infinity, then
$$
\E[F]   
\sim \frac{2}{r} n^{2-\gamma} \dvg{X}{Y}
$$
\end{reptheorem}


\section{Properties of Random Variables}
\label{app:sec:properties}


In this section, we present results on the random variables which appear in our proofs of the previous section.

\begin{lemma}
\label{lem:b_to_p}

$$
\sum_{k=0}^\nbu \E[B_k^2] 
=  \nbi +  (\nbi^2 - \nbi) \sum_{k=0}^\nbu \E[P_k^2]  
$$
\end{lemma}
    
\begin{proof}
For all $k\in [\nbu]$, we have that $B_k $ conditioned on   $(D_k, Y_k)$ follows a binomial distribution of parameters $(P_k, \nbi)$, with $P_k \triangleq \int_{Y_k}^{Y_k+D_k} f_X(x) dx$.
So, by the tower rule,

\begin{align*}
\sum_{k=0}^\nbu \E[B_k^2] 
&= \sum_{k=0}^\nbu \E [\E[B_k^2 | D_k, Y_k] ]  \\
&= \sum_{k=0}^\nbu \E [\nbi P_k + (\nbi^2 - \nbi) P_k^2]  \\
&= \sum_{k=0}^\nbu \nbi \E[P_k] + (\nbi^2 - \nbi) \sum_{k=0}^\nbu  \E[P_k^2] \\
&=  \nbi \E[ \sum_{k=0}^\nbu P_k ] + (\nbi^2 - \nbi) \sum_{k=0}^\nbu  \E[P_k^2] \\
&=  \nbi \E[1] +  (\nbi^2 - \nbi) \sum_{k=0}^\nbu \E[P_k^2] \\
&=  \nbi +  (\nbi^2 - \nbi) \sum_{k=0}^\nbu \E[P_k^2]  
\end{align*}

\end{proof}


\begin{replemma}{lem:expectation_pk2}
$\forall \beta \in (0.5, 1)$, when $\nbu$ goes to infinity,
$$
\E\left[\sum_{k=0}^\nbu P_k^2 \right]
= \sum_{k=0}^\nbu \E[(D_k f_X(Y_k))^2 ] + O \left(\frac{1}{\nbu^{2\beta}}\right) \\
$$
\end{replemma}

\begin{proof}


Lemma \ref{lem:proba_asymp} gives us
\begin{equation}
\label{lin:gm}
\mathcal{E}
\Rightarrow \forall k \in [\nbu], |P_k - D_k f_X(Y_k) | 
\leq \frac{c}{\nbu^{2\beta}} 
\triangleq \bound
\end{equation}

We have assumed $\beta > 0.5$, which implies $\bound = o\left(\frac{1}{\nbu}\right)$. \\
Then, for all $\nbu$, conditioned on $\mathcal{E}$, we have 

\begin{alignat*}{3}
D_k f_X(Y_k) - \bound
& \leq  P_k
&& \leq D_k f_X(Y_k) + \bound  \\
P_k D_k f_X(Y_k) - P_k \bound
& \leq  P_k^2
&& \leq  P_k D_k f_X(Y_k) + P_k \bound  \\
\sum_{k=0}^\nbu ( D_k f_X(Y_k) - \bound) D_k f_X(Y_k) - \bound
& \leq \sum_{k=0}^\nbu  P_k^2
&& \leq  \sum_{k=0}^\nbu (D_k f_X(Y_k) +\bound) D_k f_X(Y_k) + \bound  \\
\sum_{k=0}^\nbu ( D_k f_X(Y_k))^2 - \bound \sum_{k=0}^\nbu D_k f_X(Y_k) - \bound
& \leq \sum_{k=0}^\nbu  P_k^2
&& \leq \sum_{k=0}^\nbu ( D_k f_X(Y_k))^2 + \bound \sum_{k=0}^\nbu D_k f_X(Y_k) + \bound \\
 \sum_{k=0}^\nbu ( D_k f_X(Y_k))^2 - \bound \sum_{k=0}^\nbu (P_k + \bound) - \bound
& \leq \sum_{k=0}^\nbu  P_k^2
&& \leq \sum_{k=0}^\nbu ( D_k f_X(Y_k))^2 + \bound \sum_{k=0}^\nbu (P_k + \bound) + \bound \\
 \sum_{k=0}^\nbu ( D_k f_X(Y_k))^2 - \bound  (1 + \nbu \bound) - \bound
& \leq \sum_{k=0}^\nbu  P_k^2
&& \leq \sum_{k=0}^\nbu (D_k f_X(Y_k))^2 + \bound (1 + \nbu \bound) + \bound   \\
\end{alignat*}
where we used equation \eqref{lin:gm} several times and the property $\sum_{k=0}^\nbu P_k=1$. \\

So, for all $\nbu$, conditional on $\mathcal{E}$, we have 

$$
- \bound (2 + \nbu \bound)
\leq   \sum_{k=0}^\nbu  P_k^2  -  \sum_{k=0}^\nbu ( D_k f_X(Y_k))^2 
\leq \bound (2 + \nbu \bound)
$$

which implies 
$$
- \bound (2 + \nbu \bound)
\leq   \E[\sum_{k=0}^\nbu  P_k^2 |\mathcal{E}]  -  \E[ \sum_{k=0}^\nbu ( D_k f_X(Y_k))^2 | \mathcal{E}] 
\leq \bound (2 + \nbu \bound)
$$

So, using Lemma \ref{lem:conditional_expectation} to remove the conditioning, and letting $c_X = \max (f_X)$ 
\begin{align*}
- \bound (2 + \nbu \bound) - (1- \prob(\mathcal{E}))  - (\nbu+1) (1-\prob(\mathcal{E})) c_X^2
\leq \E[\sum_{k=0}^\nbu  P_k^2 ]  -  \E[ \sum_{k=0}^\nbu ( D_k f_X(Y_k))^2 ]  \\
\E[\sum_{k=0}^\nbu  P_k^2 ]  -  \E[ \sum_{k=0}^\nbu ( D_k f_X(Y_k))^2 ]
\leq \bound (2 + \nbu \bound) + (1- \prob(\mathcal{E}))  + (\nbu+1) (1-\prob(\mathcal{E})) c_X^2\\
\end{align*}

This gives us
\begin{align*}
| \E[\sum_{k=0}^\nbu  P_k^2 ]  -  \E[ \sum_{k=0}^\nbu ( D_k f_X(Y_k))^2 ] |
& \leq \bound (2 + \nbu \bound) + (1- \prob(\mathcal{E}))  + (\nbu+1) (1-\prob(\mathcal{E})) c_X^2 \\ 
& = O(\bound) \\
& = O \left( \frac{1}{\nbu^{2 \beta}} \right)
\end{align*}

Because, from Lemma \ref{lem:all_bins_variation}, we have 
$ 1 - \prob(\mathcal{E})  
\leq  (\nbu+1) e^{-\nbu^{1-\beta}}$

\end{proof}

Lemma \ref{lem:expectation_pk2} makes the expression $\E[(D_k f_X(Y_k))^2]$ appear.
For all $k$, the left extremity ($Y_k$) of the bin  and the length ($D_k$) of the bin are not independent random variables.
For this reason, we will need information on their joint distribution.
For this, we derive the conditional distribution of $D$ given $Y$ and its first two moments in Lemmas \ref{lem:d_cond_y}, \ref{lem:d1_cond_y}, \ref{lem:d2_cond_y}.

\begin{lemma}
\label{lem:d_cond_y}
For all $y \in (0,1],$ the conditional density of $D$ given $Y=y$ is
$$
f_{D|Y=y>0}(x) = (m-1)(1-x)^{m-2} \mathbb{1}(x < 1-y) + y^{m-1} \delta_{1-y}(x)
$$
For the case $y=0$, we have
$$
f_{D|Y=0}(x) = m (1-x)^{m-1}
$$
\end{lemma}

\begin{proof}
Let us start by the case $y=0$.
Conditional on $Y_K=0$, we know that $K=0$, so $D \triangleq D_K = D_0$.
$D_0$ is a $Beta(1,m)$, and is independent from $K$.
So $D | Y=0$ is a $Beta(1,m)$.
\emph{i.e.} $f_{D|Y=0}(x)= m(1-x)^{m-1}$.

We now consider $y \in (0,1]$. 
Let $d \in (0,1]$. \\
Then the probability $\prob(D>d | Y=y)$ is zero if $d \geq 1-y$ (the bin cannot be longer than the space between its left extremity and $1$). \\
If $d < 1-y$, the event $(D>d)$
is equivalent to having no threshold in the interval $[y, y+d]$
So $\prob(D>d | Y=y)$ is the probability of having no threshold in $[y, y+d]$ given that there is a threshold at $y$.
The sampling of the thresholds is \emph{iid} uniform in $[0,1]$.
Therefore, the probability for one threshold to not be in the interval is $(1-d)$, and the probability for no threshold to be there is $(1-d)^{\nbu-1}$ (one threshold is fixed at position $y$ by the conditioning).
So, we have:

For $y > 0$, $\prob(D>d | Y=y) = (1-d)^{m-1} \mathbb{1}(d < 1-y)$

By differentiating the ccdf, we obtain the result:

$$
f_{D|Y=y}(x) = (m-1)(1-x)^{m-2} \mathbb{1}(x < 1-y) + y^{m-1} \delta_{1-y}(x)
$$

\end{proof}

\begin{lemma}
\label{lem:d1_cond_y}
$$
\forall m \in \N, \quad
\E[D | Y=0] =  \frac{1}{m+1}
$$
$$
\forall m \geq 1, \quad
\E[D | Y=y>0] =  \frac{1}{m} - \frac{y^m}{m} 
$$
\end{lemma}

\begin{proof}

Recall that by definition, $D = D_K$, where $K$ is uniform on $[\nbu] \cup \{0\}$ and that $Y_0 = 0$ by convention. 
This means that $Y_K=0$ implies $K=0$, ie $D_K = D_0$, which follows a $Beta(1, \nbu)$ (Lemma \ref{lem:diffunif}).

So

$$
\E[D | Y=0] =  \frac{1}{m+1}
$$

In addition, using Lemma \ref{lem:d_cond_y},
    
\begin{align*}
\E[D^k | Y=y>0] 
&= \int_0^1 x^k(m-1)(1-x)^{m-2} (\mathbb{1}(x<1-y) + y^{m-1} \delta_{1-y}(x) )dx \\
&= (m-1) \int_0^{1-y} x^k (1-x)^{m-2}  dx + y^{m-1} \int_0^1 x^k \delta_{1-y}(x) )dx \\
&= (m-1) \int_0^{1-y} x^k (1-x)^{m-2}  dx + y^{m-1} (1-y)^k \\
\end{align*}

In particular, for $k=1$, using integration by parts,

\begin{align*}
\E[D | Y=y>0] 
&= (m-1) \int_0^{1-y} x (1-x)^{m-2}  dx + y^{m-1} (1-y) \\
&= (m-1)  \left[ x \frac{- (1-x)^{m-1}}{m-1} \right]_0^{1-y}  -  (m-1)\int_0^{1-y} \frac{- (1-x)^{m-1}}{m-1} dx
+ y^{m-1} (1-y) \\
&=  -(1-y)y^{m-1} - \frac{y^m}{m} + \frac{1}{m} + y^{m-1} (1-y) \\
&=   \frac{1}{m} - \frac{y^m}{m} 
\end{align*}

\end{proof}

\begin{lemma}
\label{lem:d2_cond_y}

$$
\forall \nbu \in \N, \quad
\E[D^2 | Y = 0] 
= \frac{2}{ (\nbu+1)(\nbu+2)}
$$

$$
\forall \nbu \geq 1, \quad 
\E[D^2 |Y=y>0] 
=  \frac{2}{\nbu(\nbu+1)} \left( 1 - y^{(\nbu+1)} - (\nbu+1) y^{\nbu}(1-y) \right)
$$

\end{lemma}

\begin{proof}

Recall that by definition, $D = D_K$, where $K$ is uniform on $[\nbu]\cup \{0\}$ and that $Y_0 = 0$ by convention. 
This means that $Y_K=0$ implies $K=0$, ie $D_K = D_0$, which follows a $Beta(1, \nbu)$ (Lemma \ref{lem:diffunif}).
So

$$
\E[D_K^2 | Y_K = 0] 
= \E[D_0^2]
= \frac{2}{ (\nbu+1)(\nbu+2)}
$$

We now treat the case $Y = y > 0$. 
We can prove it in two different ways.

\paragraph{Method 1}:

Lemma \ref{lem:d_cond_y} gives us 
$
f_{D|Y=y>0}(x) = (m-1)(1-x)^{m-2} \mathbb{1}(x < 1-y) + y^{m-1} \delta_{1-y}(x)
$.
So 

\begin{align*}
\E[D^2 | Y=y]
&= \int_0^1 x^2 f_{D | Y=y}(x) dx \\
&= \int_0^1 x^2 ((m-1)(1-x)^{m-2} \mathbb{1}(x < 1-y) + y^{m-1} \delta_{1-y}(x)) dx \\
&= (m-1) \int_0^{1-y} x^2 (1-x)^{m-2} dx + \int_0^1 x^2 y^{m-1} \delta_{1-y}(x) dx \\
&= (m-1) \int_0^{1-y} x^2 (1-x)^{m-2} dx + (1-y)^2  y^{m-1}  \\
\end{align*}
with

\begin{align*}
\int_0^{1-y} x^2 (1-x)^{m-2} dx  
&=   \left[x^2 \frac{-(1-x)^{m-1}}{m-1} \right]_0^{1-y} 
- \int_0^{1-y} 2 x \frac{-(1-x)^{m-1}}{m-1} dx\\ 
&=   \left[x^2 \frac{-(1-x)^{m-1}}{m-1} \right]_0^{1-y}
- 2 \left[   x \frac{(1-x)^{m}}{(m-1)m} \right]_0^{1-y}   
- 2 \int_0^{1-y}  \frac{-(1-x)^{m}}{(m-1)m} dx\\ 
&=   -(1-y)^2 \frac{y^{m-1}}{m-1}
- 2   (1-y) \frac{y^{m}}{(m-1)m}   
- 2 \frac{y^{m+1} - 1 }{(m-1)m(m+1)}  \\ 
\end{align*}

so finally

\begin{align*}
\E[D^2 | Y=y]
& =   - 2   (1-y) \frac{y^{m}}{m}   
+ 2 \frac{1 - y^{m+1} }{m(m+1)}  \\ 
& =   \frac{2}{m(m+1)} \left(   -(m+1) (1-y) y^{m}
+  1 - y^{m+1} \right)  \\ 
& =   \frac{2}{m(m+1)} \left(1  - y^{m+1} -(m+1) (1-y) y^{m}
   \right)  \\ 
\end{align*}

\paragraph{Method 2}: 

The thresholds are sampled \emph{iid} uniformly in $[0,1]$ (because $f_Y=1$).
In particular, each one of the (unordered) thresholds has a probability $y$ of being smaller than $y$, independently of each other. 
We have $\nbu-1$ such thresholds, because there are a total of $\nbu$ thresholds, including $Y_K=y$ itself.
Therefore, the distribution of $K$ conditional on $Y_K=y\neq 0$ is one plus a binomial distribution of parameters $(y, \nbu-1)$ .\\
\emph{ie} $\forall 0<y\leq1, \quad K - 1 | Y_K = y \sim Bin(y, \nbu-1)$. \\

In addition, conditioned on $K=k$, we know that we have exactly $\nbu +1 - k$ thresholds greater than $Y_k$, and that the distribution of these thresholds is uniform \emph{iid} on $[Y_k, 1]$.
This means that the distance $D_k = Y_{k+1} - Y_k$ follows a Beta distribution rescaled on the interval $[Y_k, 1]$.

So, formally, $\frac{D_k}{1-y}  | Y_k = y \sim Beta(1, \nbu-k)$
(with the convention $Beta(1, 0)$ is a Dirac in $1$). \\

In general, if $Z \sim Beta(1,a)$, then $\E[Z^2] = \frac{2}{(1+a)(2+a)}$. 
This gives us
\begin{align*}
\E[D_K^2 | Y_K = y, K] 
&= 2(1-y)^2 \E \left[ \frac{1}{(1 + \nbu - K)(2 + \nbu - K)} \middle| Y_K=y, K \right] 
\end{align*}
\begin{align*}
\E[D_K^2 | Y_K = y] 
&= \E[ \E[D_K^2 | Y_K = y, K] ] \\
&= 2(1-y)^2 \E \left[ \frac{1}{(1 + \nbu - K)(2 + \nbu - K)} \middle| Y_K=y \right] \\
&= 2(1-y)^2 \sum_{k=1}^\nbu  \frac{1}{(1 + \nbu - k)(2 + \nbu - k)} \prob(K=k | Y_K =y) \\
&= 2(1-y)^2 \sum_{k=1}^\nbu  \frac{1}{(1 + \nbu - k)(2 + \nbu - k)} \binom{\nbu-1}{k-1} y^{k-1}(1-y)^{m-k} \\
&= 2 \sum_{k=1}^\nbu  \frac{\frac{1}{(\nbu+1)\nbu}(\nbu+1)!}{(k-1)! (\nbu-k+2)!}  y^{k-1}(1-y)^{\nbu-k+2} \\
&=  \frac{2}{(\nbu+1)\nbu} \sum_{k=1}^{\nbu}  \binom{\nbu+1}{k-1}  y^{k-1}(1-y)^{\nbu-k+2} \\
&=  \frac{2}{(\nbu+1)\nbu} \sum_{k=0}^{\nbu-1}  \binom{\nbu+1}{k}  y^{k}(1-y)^{\nbu-k+1} \\
&=  \frac{2}{(\nbu+1)\nbu} \left( \left( \sum_{k=0}^{\nbu+1}  \binom{\nbu+1}{k}  y^{k}(1-y)^{\nbu+1-k} \right) - y^{\nbu+1} - (\nbu+1) y^{\nbu}(1-y) \right)\\
&=  \frac{2}{(\nbu+1)\nbu} \left( 1 - y^{\nbu+1} - (\nbu+1) y^{\nbu}(1-y) \right)
\end{align*}

\end{proof}

Because we made the assumption that $f_Y=1$, the user thresholds are \emph{iid} uniform random variables.
The following lemma shows that for all $k$, the bin length $D_k$ follows a $Beta(1, \nbu)$.

\begin{lemma} 
\label{lem:diffunif}
Let $\rv_{i \in [m]}$ be iid uniform random variables on $[0, 1]$. Let $\rv_{(i), i \in [m]}$ be the order statistics of $\rv$. 
Then \\
$ \forall i \in [m-1]$,
$ \rv_{(i+1)} - \rv_{(i)} \sim Beta(1, m), \quad$  i.e. $f_{\rv_{(i+1)} - \rv_{(i)}}(x) 
    = m (1-x)^{m-1} \mathbb{1}(x \in [0,1])
    $ 
\end{lemma}

\begin{proof} 
In order to generate \emph{iid} uniform random variables on $[0,1]$, one can sample $m+1$ uniform random variables on a unit circle and then pick one of them to be the start of the $[0,1]$ interval (the $0$ point). 
This way, it is clear that $\rv_{(1)} - 0$ has the same distribution as $\rv_{(i+1)} - \rv_{(i)}$ for all $i \in [m-1]$.\\
It is well-known that the minimum of $m$ uniform random variables on $[0,1]$ follows a $Beta(1,m)$, so we also have $\rv_{(i+1)} - \rv_{(i)} \sim Beta(1,m)$. \\
(Proof adapted from  \href{https://math.stackexchange.com/questions/68749/difference-of-order-statistics-in-a-sample-of-uniform-random-variables}{the one given by  Liran Katzir on Stack Exchange})
\end{proof}

\begin{lemma}
\label{lem:bound_beta}
Let $m \in \N$, let $Z \sim Beta(1, m)$.
For all $\beta \in [0,1], a >0$,

$$
\prob\left(Z \leq \frac{a}{\nbu^\beta}\right) 
\geq 1 - e^{-a\nbu^{1-\beta}} 
$$
\end{lemma}

\begin{proof}

$$
\prob(Z \geq \frac{a}{\nbu^\beta}) 
= \int_{\frac{a}{\nbu^\beta}}^1 m (1-x)^{m} dx
=   (1-\frac{a}{\nbu^\beta})^m
$$

Using the fact that $\forall x \in \R, \log(1-x) \leq -x$,
$$
\log \left((1-\frac{a}{\nbu^\beta})^m \right)
= m \log \left(1-\frac{a}{m^\beta} \right)
\leq - m \frac{a}{m^\beta}
= -am^{1 - \beta}
$$
which implies
$
\prob(Z \geq \frac{a}{\nbu^\beta}) 
= (1-\frac{a}{\nbu^\beta})^m 
\leq e^{-am^{1-\beta}}
$.
\end{proof}


\section{Uniform upper bound on the length of the bins}
\label{app:sec:event}

Because we assume $f_X$ to be $c$-Lipschitz, we understand that, when $m$ grows large, $f_X$ will be almost constant over each bin interval $[Y_k, Y_{k+1}]$. In this section, we formalize this intuition.
For this, we define a probabilistic event $\mathcal{E}$ under which the length of all bins is upper bounded.
In Lemma \ref{lem:all_bins_variation}, we show that this event happens with probability going exponentially to $1$.

\begin{repdefinition}{def:event}
For all $ \beta \in (0,1),  \nbu \in \N$, we define the event $\mathcal{E}(\beta, \nbu)$ as follows:

$$
\mathcal{E}(\beta, \nbu) 
\triangleq \left( \bigcap_{k=0}^\nbu \left( D_k \leq  \frac{1}{\nbu^\beta} \right) \right) 
$$
\end{repdefinition}

\begin{lemma}
\label{lem:all_bins_variation}

(i) $\forall \beta \in (0,1), \forall \nbu \in \N$
$$
\prob \left(\mathcal{E}(\beta, \nbu) \right) 
\geq 1 - (\nbu +1) e^{-\nbu^{1 - \beta }}
$$

(ii) $\mathcal{E}(\beta, \nbu)$ implies that the variation of $f_X$ inside any bin is smaller than $\frac{c}{\nbu^\beta}$, ie 
\begin{align*}
\mathcal{E}(\beta, \nbu) 
&  \Rightarrow \left( \bigcap_{k=0}^\nbu \left( \forall x \in [Y_k, Y_{k+1}], |f_X(x) - f_X(Y_k)| \leq  \frac{c}{\nbu^\beta} \right) \right) \\
 & = \left( \sup_{k \in [\nbu]} 
\sup_{x \in [Y_k, Y_{k+1}]} |f_X(x) - f_X(Y_k)| 
\leq  \frac{c}{\nbu^\beta} \right) 
\end{align*}

where $c$ is the Lipschitz constant of $f_X$.

\end{lemma}

\begin{proof}
(i) Let $\nbu \in \N$. 
Because the thresholds are \emph{iid} uniform on $[0,1]$, we have that $\forall k \in [\nbu] \cup \{0\}, D_k \sim Beta(1, \nbu)$ (this is a property of the order statistics of the uniform distribution, see Lemma \ref{lem:diffunif} below in Appendix). \\

Lemma \ref{lem:bound_beta} gives us
$$
\forall k \in [\nbu], \prob(D_k \leq \frac{1}{\nbu^\beta}) 
\geq 1 - e^{-\nbu^{1-\beta}}
$$

So, by union bound on all $k$ in $\N$, 
\begin{align*}
\prob(\forall k \in [\nbu], D_k  \leq \frac{1}{\nbu^\beta}) 
& \geq 1 -  \sum_{k=0}^\nbu e^{-\nbu^{1-\beta}} \\
& = 1 -  (\nbu +1) e^{-\nbu^{1-\beta}}
\end{align*}

(ii) By definition, $\mathcal{E}(\beta, \nbu) \Rightarrow \left(\forall k \in [\nbu], D_k \leq \frac{1}{\nbu^\beta} \right)$. \\

Coupled with the fact that $f_X$ is $c$-Lipschitz, we obtain
$$
\forall k \in [\nbu], \forall x \in [Y_k, Y_{k+1}], |f_X(x) - f_X(Y_k)| \leq c D_k
$$

So by combining the two, we have 

$$
\mathcal{E}(\beta, \nbu)
\Rightarrow \left( \forall k \in [\nbu], \forall x \in [Y_k, Y_{k+1}], |f_X(x) - f_X(Y_k)| \leq  \frac{c}{\nbu^\beta}  \right) \\
$$

\end{proof}

Because the density $f_X$ of the items scores is not constant on $[0,1]$, the conditional probability of an item being in bin $k$ conditioned the user thresholds does not depend only on the length $D_k$ of the bin, but also on its position in the $[0,1]$ interval. \\
We first provide some notation for this conditional probability.

We now provide an approximation of $P_k$, by conditioning on event $\event$.

\begin{lemma}
\label{lem:proba_asymp} 
For all $i \in [\nbi]$, the (conditional) probability of item $i$ being in bin $k$ can be approximated by $D_k f_X(Y_k)$. \\

$\forall \in \N,  \forall \beta \in (0,1),$
$$ 
\mathcal{E}(\beta, \nbu) 
\Rightarrow 
\left(
\forall k \in [\nbu], 
\quad |P_k - D_k f_X(Y_k)| 
\leq  \frac{ c}{\nbu^{2\beta}} 
\right)
$$ 

where $c$ is the Lipschitz constant of $f_X$.

\end{lemma}

\begin{proof}
Using Lemma \ref{lem:all_bins_variation}, conditional on the event $\mathcal{E}(\beta, \nbu)$, we have, for all $k\in [\nbu] \cup \{0\}$

\begin{alignat*}{3}
\forall x \in [Y_k, Y_{k+1}], \quad
f_X(Y_k) - \frac{c}{\nbu^\beta}
& \leq    f_X(x) 
&&  \leq f_X(Y_k) + \frac{c}{\nbu^\beta} \\
D_k (f_X(Y_k) - \frac{c}{\nbu^\beta})
& \leq \int_{Y_k}^{Y_k+D_k} f_X(x) dx 
&& \leq  D_k (f_X(Y_k) + \frac{c}{\nbu^\beta}) \\
D_k (f_X(Y_k) - \frac{c}{\nbu^\beta})
& \leq P_k
&& \leq  D_k (f_X(Y_k) + \frac{c}{\nbu^\beta})  \\
- D_k  \frac{c}{\nbu^\beta}
& \leq P_k - D_k f_X(Y_k)
&& \leq  D_k \frac{c}{\nbu^\beta} 
\end{alignat*}

and $\mathcal{E}(\beta, \nbu) \Rightarrow D_k \leq \frac{1}{\nbu^\beta}$, so
$$ 
\mathcal{E}(\beta, \nbu) \Rightarrow
|P_k - D_k f_X(Y_k)| 
\leq  D_k   \frac{c}{\nbu^\beta}
\leq   \frac{ c}{\nbu^{2\beta}} 
$$ 
\end{proof}

The probability of event $\mathcal{E}$ goes to 1 as $\nbu$ goes to infinity (Lemma \ref{lem:all_bins_variation}).
Therefore, it seems intuitive that an expectation conditional on this event should not be too different from the unconditional expectation.
We now prove a lemma that formalizes this intuition.
This result will be useful to approximate unconditional expectations by using properties valid under $\mathcal{E}$.

\section{Influence of $f_X$ and $f_Y$}
\label{app:sec:divergence}

\subsection{Reducing the problem to the $f_Y=1$ case}
\label{app:sec:reduction}

As explained in Section \ref{sec:main_theorems}, we can assume without loss of generality that $f_Y=1$.
We formally prove it in this section.
We start by recalling the notation defined in Appendix \ref{app:sec:definitions}, but without the assumption $f_Y=1$.
Then, we show that we can rescale everything in order to obtain an equivalent setting in which we have $f_Y=1$. 

\begin{itemize}
    \item $f_X$: the density of the item scores on $[0,1]$. Scores are \emph{iid} on the interval.
    \item $f_Y$: the density of the user thresholds on $[0,1]$.
    Thresholds are \emph{iid} on the interval.
    \item $\forall i \in [\nbi], X_i$ is the (unordered) score of item $i$. $X_i$s are \emph{iid} of density $f_X$.
    \item $\forall k \in [\nbu], Y_k$ is the $k$-th smallest  threshold, \emph{i.e.} the $k$-th order statistic of $\nbu$ \emph{iid} random variables of density $f_Y$ (with convention $Y_0 = 0$, and $Y_{k+1} =1$).
    \item $\forall k \in [\nbu] \cup \{0\}, B_k  \triangleq |\{i \in [\nbi] | X_i \in [Y_k, Y_{k+1}] \} |$, the number of items in bin $k$. 
    \item $B \triangleq B_K, K \sim \mathcal{U}([\nbu] \cup \{0\})$, the number of items in a random bin, chosen uniformly at random among the bins. 
    We define in the same way $Y \triangleq Y_K$ and $D \triangleq D_K $, a random threshold and the length of a random bin.
    \item $F$: the value of the MSF.
\end{itemize}

$f_X$ and $f_Y$ are assumed to take non-zero values on $(0,1)$. \\

In order to simplify the analysis, we define the \emph{rescaled} thresholds and item scores, by composing everything by the \emph{cdf} $F_Y$.
Thanks to this manipulation, the \emph{rescaled} thresholds are uniform on $[0,1]$, and the item scores remain \emph{iid}.

Formally, we define
$\forall i \in [\nbi], \quad X'_i \triangleq F_{Y}(X_i)$  and
$\forall k \in [\nbu], \quad Y'_k \triangleq F_{Y}(Y_k)$ \\
So,
$$
\forall x,y \in [0,1]
,\quad
f_{X'}(x) = \frac{f_{X}(F_{Y}^{-1}(x))}{f_{Y}(F_{Y}^{-1}(x))}
,\quad
f_{Y'}(y) = \frac{f_{Y}(F_{Y}^{-1}(y))}{f_{Y}(F_{Y}^{-1}(y))}
= 1
$$

We can define all the corresponding \emph{rescaled} objects:
\begin{itemize}
    \item $f_{X'}$: the density of the \emph{rescaled} item scores on $[0,1]$. 
    \item $f_{Y'}$: the density of the \emph{rescaled}  user thresholds on $[0,1]$. 
    As proven above, $f_{Y'}=1$.
    \item $\forall i \in [\nbi], X'_i \triangleq F_{Y}(X_i)$ is the (unordered) \emph{rescaled} score of item $i$.
    $X'_i$s are \emph{iid} of density $f_{X'}$.
    \item $\forall k \in [\nbu], Y'_k \triangleq F_{Y}(Y_k)$ is the $k$-th smallest  \emph{rescaled} threshold .
    \item $\forall k \in [\nbu] \cup \{0\}, D'_k \triangleq Y'_{k+1}-Y'_k$ the \emph{rescaled} length of bin number $k$ (with $Y'_0 = F_{Y}(Y_0) = 0 $ and $Y'_{\nbu+1}= F_{Y}(Y_{\nbu+1}) = 1$).
    \item $\forall k \in [\nbu] \cup \{0\}, B'_k    \triangleq |\{i \in [\nbi] | X'_i \in [Y'_k, Y'_{k+1}] \} | = B_k$, the number of items in bin $k$. 
    \item $B' \triangleq B'_K = B, K \sim \mathcal{U}([\nbu] \cup \{0\})$, the number of items in a random bin, chosen uniformly at random among the bins. 
    We define in the same way $Y' \triangleq Y'_K$ and $D' \triangleq D'_K $, a random \emph{rescaled} threshold and the length of a random \emph{rescaled} bin.
    Note that we have $Y' = F_{Y}(Y)$, because $\forall k \in [\nbu], \quad Y'_k \triangleq F_{Y}(Y_k)$.
    \item $F' = F$: the value of the MSF.
\end{itemize}

Because this transformation is increasing, the order of the thresholds and the items is preserved.
As a consequence, the MSF remains unchanged.
This means that we can reason under the assumption that $f_Y=1$ without loss of generality, because the analysis of the MSF given any $(f_X, f_Y)$ is equivalent to the analysis given $(f_{X'}, 1)$, as defined above.
Under this assumption, our results are expressed in function of $\E[f_X(Y)^2]$, and converted to the full quadratic divergence of the general case using Lemma \ref{lem:rescale}.

\subsection{Properties of the quadratic divergence}

In the following lemma, we show that the $\E[f_X(Y)^2]$ term becomes the $\E\left[ \frac{f_X(Y)^2}{f_Y(Y)^2} \right]$ term of Theorems \ref{thm:msf_high} and \ref{thm:msf_low} for $f_Y \neq 1$.

\begin{lemma}
\label{lem:rescale}
Let $(f_{X}, f_{Y})$ be the true densities of the scores and thresholds, and $(f_{X'}, 1)$ be their rescaled counterparts. 
Then
$$
\E[f_{X'}(Y')^2 ]
= \E\left[ \left( \frac{f_{X}(Y)}{f_{Y}(Y)} \right)^2  \right]
$$
\end{lemma}

\begin{proof}

we have

$\forall i \in [\nbi], X'_i \triangleq F_{Y}(X_i)$ and
$\forall k \in [\nbu], Y'_k \triangleq F_{Y}(Y_k)$, so 

$X'= F_{Y}(X)$ and 
$Y'= F_{Y}(Y)$

So, 
$$
\forall x \in [0,1], \quad
f_{X'}(x) = \frac{f_{X}(F_{Y}^{-1}(x))}{f_{Y}(F_{Y}^{-1}(x))}
$$

$$
\Rightarrow
f_{X'}(Y') = \frac{f_{X}(F_{Y}^{-1}(Y'))}{f_{Y}(F_{Y}^{-1}(Y'))}
= \frac{f_{X}(Y)}{f_{Y}(Y)}
$$

$$
\Rightarrow
\E[f_{X'}(Y')^2 ]
= \E\left[ \left( \frac{f_{X}(Y)}{f_{Y}(Y)} \right)^2  \right]
$$
\end{proof}

As discussed in Section \ref{sec:main_theorems}, the intuition suggests that the ideal case is when the score and item distributions are the same.
In the following lemma, we show that the divergence term is indeed minimal \emph{iif} $f_X=f_Y$.

\begin{lemma}
\label{lem:min_div}
$$ \E\left[\frac{f_X(Y)^2}{f_Y(Y)^2} \right] 
\geq 1$$

and $\E\left[\frac{f_X(Y)^2}{f_Y(Y)^2} \right] 
=  1$ \emph{iif} $f_X(Y) = f_Y(Y)$ \emph{a.s.}.
\end{lemma}

\begin{proof}
As explained at the beginning of the Section, we can prove the results with the assumption $f_Y=1$ without loss of generality.
By convexity, we have:
$$
\E[f_X(Y)^2] 
\geq \E[f_X(Y)]^2 
= \left(\int_0^1 f_X(y) dy \right)^2
= 1
$$

with equality \emph{iif} $f_X=1$, \emph{i.e.} $f_X=f_Y$.
\end{proof}

Finally, we look at the particular case where $X$ and $Y$ follow Beta distributions in order to get a closed form for this divergence term.

\begin{lemma}
\label{lem:divergence_beta}

Let $X$ and $Y$ be two Beta random variables with
$ X \sim Beta(a_X, b_X) $ and $Y \sim Beta(a_Y, b_Y)$, 
where $a_X > \frac{b_Y}{2}  > 0$, $a_X \geq \frac{b_X}{2} > 0$. 
Then we have
$$
\dvg{X}{Y}  
= \frac{\betaf(a_Y, b_Y)}{\betaf(a_X, b_X)^2}  \betaf(2a_X - a_Y, 2b_X - b_Y) 
$$ 
where $\betaf$ is the Beta function.
\end{lemma}

\begin{proof}
\begin{align*}
\dvg{X}{Y} 
& \triangleq \E_Y \left[ \left(  \frac{f_X(Y)}{f_Y(Y)}  \right)^2 \right] \\
& = \int_0^1  \left(  \frac{f_X(y)}{f_Y(y)} \right)^2    f_Y(y) dy \\
& = \int_0^1  \frac{f_X(y)^2}{f_Y(y)} dy \\
&=  \frac{\betaf(a_Y, b_Y)}{\betaf(a_X, b_X)^2} \int_0^1  \frac{y^{2(a_X-1)}(1-y)^{2(b_X-1)}}{y^{a_Y-1}(1-y)^{b_Y-1}} dy \\
&=  \frac{\betaf(a_Y, b_Y)}{\betaf(a_X, b_X)^2} \int_0^1  y^{2a_X-a_Y-1}(1-y)^{2b_X-2b_Y-1} dy \\
&=  \frac{\betaf(a_Y, b_Y)}{\betaf(a_X, b_X)^2}  \betaf(2a_X - a_Y, 2b_X - b_Y) 
\end{align*}

\end{proof}

\section{Appendix of Section \ref{sec:algorithm}, on the $\userdic$ Algorithm}
\label{app:sec:algo}

In this section, we make the distinction between $Y_u$, the threshold of user $u$, and $\threshord{k}$, the $k$-th smallest user threshold.

We recall that, as explained in Section \ref{sec:algorithm}, we split the total query cost $\sampcomp$ in the following way:

\begin{equation}
\label{eq:queries_sum}
    \sampcomp = \tsearch + \tiso + \tsplit 
\end{equation} 

where $\tsearch$, $\tiso$, and $\tsplit$ are the number of queries performed respectively during the $\bsearch$, $\bisolate$ and $\bsplit$ phases.

In the specification of Algorithm \ref{alg:userdic}, nothing prevents the algorithm from making several times the same query (in different phases).
However, it is reasonable to consider that the ratings given by the users are stored, and that asking a second time the same rating to the same user does not count as a query.

Under this assumption, we define $\tiso$ and $\tsplit$ such that we count all the queries on the items of bin $B_{k^*(\idxuser)}$ in $\tsplit$ and not in $\tiso$, where $B_{k^*(\idxuser)}$ is the bin being split at step $\idxuser$.
Note that this change of definition does not affect equation \eqref{eq:queries_sum}.
In what follows, we refer by $\bsearch_u$, $\bisolate_u$, $\bsplit_u$ and $\bclear_u$ to the execution of these phases at step $u$.\\ 

We will now upper bound the expectations of these three random variables.
For this, we use the three following lemmas:

\begin{lemma}
\label{lem:qsearchpm}
    $\E[\tsearch] \leq \nbu (\log_2(\nbi) + 1)$
\end{lemma}

\begin{lemma}
\label{lem:tisoleqtsplit}
     $\tiso \leq \tsplit$
 \end{lemma}

\begin{lemma}
\label{lem:qsplit}
$\E[\tsplit]  
\lesssim  2 n \log(m) + 2 m $ 
\end{lemma}

We restate and prove each of these three lemmas:

\begin{replemma}{lem:qsearchpm}
    $\E[\tsearch] \leq \nbu (\log_2(\nbi) + 1)$
\end{replemma}

\begin{proof}
For each user, we perform a binary search on the current set of non-empty bins. There cannot be more non-empty bins than the number of items $\nbi$.
Binary search over $\nbi$ bins takes at most $\lceil \log_2(\nbi) \rceil$ queries.
Therefore, each user can do at most $\lceil \log_2(\nbi) \rceil$ ratings during the $\bsearch$ phase of user $u$.

This gives us:

\begin{align*}
    \tsearch
    & \leq \sum_{u=1}^\nbu \lceil \log_2(\nbi) \rceil \\
    & \leq \nbu (\log_2(\nbi) + 1) \\
\end{align*}
and in particular
$$
\E[\tsearch] \leq \nbu (\log_2(\nbi) + 1)
$$

\end{proof}

\begin{replemma}{lem:tisoleqtsplit}
 $\tiso \leq \tsplit$
\end{replemma}

\begin{proof}
    For a given user $\idxuser$, let $\tiso_\idxuser$ (\emph{resp.} $\tsplit_\idxuser$) be the number of ratings performed by $\idxuser$ during $\bisolate$ (\emph{resp.} $\bsplit$). 
    
Suppose $\bin_\idxright$ is returned by $\bisolate_\idxuser$.
In reality, in $\bisolate_\idxuser$, we rate some (or all) items in $\bin_\idxleft$ and some in $\bin_\idxright$. 
However, as stated at the beginning of the section, for accounting purposes, we count the ratings of $\bin_\idxright$ made during $\bisolate$ as part of $\bsplit$ (so we bound them in Lemma \ref{lem:qsplit}).
In $\bsplit$, we anyway rate the rest of $\bin_r$.
So, we simply say that we rate (some or all) items of $\bin_\idxleft$ in $\bisolate_\idxuser$ and all of $\bin_\idxright$ in $\bsplit_\idxuser$.

Now there are two cases: either $\bisolate_\idxuser$ stopped after rating all items from $\bin_\idxleft$ or it stopped before.

In the first case, the construction of $\bisolate$ implies that $|\bin_\idxleft| \leq |\bin_\idxright|$.  
By assumption, we have $\tsplit_\idxuser = |\bin_\idxright|$. \\
This gives $ \tiso_\idxuser \leq |\bin_\idxleft| \leq |\bin_\idxright| = \tsplit_\idxuser $.

In the second case, the construction of $\bisolate$ implies that we rated an equal number of items of $\bin_\idxleft$ and $\bin_\idxright$ during $\bisolate$.
However, the ones of $\bin_\idxright$ are counted in $\tsplit_\idxuser$.
This gives us $ \tiso_\idxuser \leq \tsplit_\idxuser $.

Consequently, in all cases,   we have $\tiso_\idxuser \leq \tsplit_\idxuser$. \\
As this is true for every $\idxuser$, this directly yields $\tiso  = \sum_{\idxuser=1}^\nbu \tiso_\idxuser \leq \sum_{\idxuser=1}^\nbu \tsplit_\idxuser = \tsplit$.
\end{proof}

\begin{replemma}{lem:qsplit}
$\E[\tsplit]  
\lesssim  2 n \log(m) + 2 m $ 
\end{replemma}

\begin{proof}[Idea of proof]

We use some inexact approximations to prove the inequality (hence the $\lesssim$ in the statement).

At a given step $\idxuser$ of the algorithm (after user $\idxuser - 1$), for a given threshold  index $\idxbin$ we define:
\begin{itemize}
    \item $D_\idxbin(\idxuser) \triangleq \threshord{\idxbin+1}(\idxuser) - \threshord{\idxbin}(\idxuser)$ the \emph{length} of bin $\idxbin$, \emph{i.e.} the distance between two consecutive thresholds.
    \item $\score^+_\idxbin(\idxuser)$, the smallest item score bigger than  $\threshord{\idxbin+1}(\idxuser)$. 
    \item $\score^-_\idxbin(\idxuser)$, the biggest item score smaller than $\threshord{\idxbin}(\idxuser)$.
    \item $D_\idxbin^+(\idxuser) \triangleq \score^+_\idxbin(\idxuser) - \score^-_\idxbin(\idxuser)$.
    \item $d^-_\idxbin(\idxuser) 
    \triangleq \threshord{\idxbin}(\idxuser) - \score^-_\idxbin(\idxuser)$.
    \item $d^+_\idxbin(\idxuser) 
    \triangleq \score^+_\idxbin(\idxuser) - \threshord{\idxbin+1}(\idxuser) $.
\end{itemize}

This new notation is illustrated on Figure \ref{fig:queriespmdefs}. \\

\begin{figure}[!ht]
\centering
  \centering
  \includegraphics[width=\linewidth]{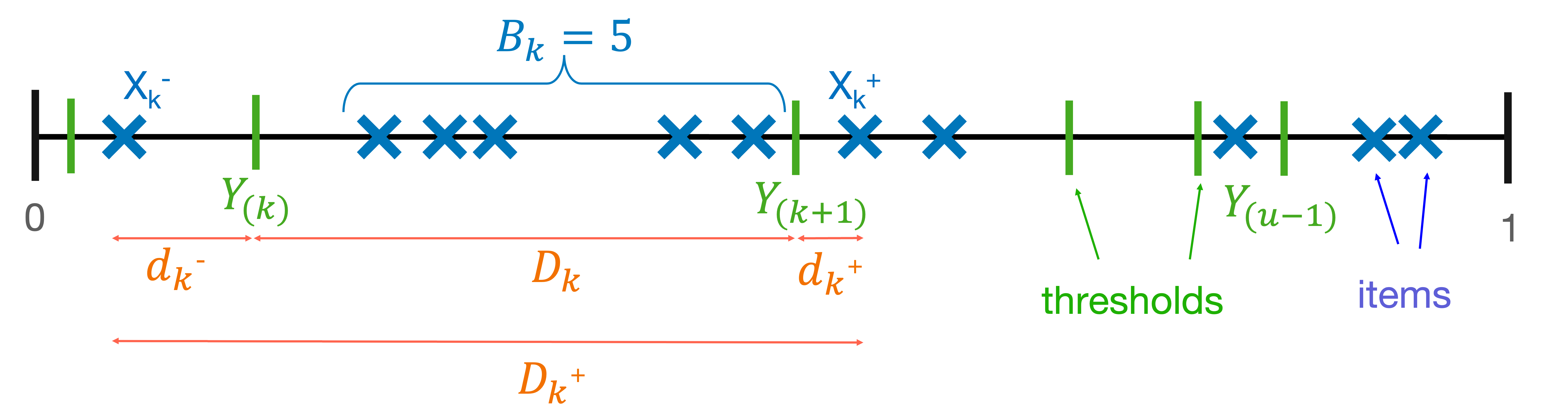}
  \caption{Illustration of the definitions for $\idxuser$ fixed. We omit the dependency in $\idxuser$ on the figure for readability. 
  All random variables here depend on $\idxuser$.
  }
  \label{fig:queriespmdefs}
\end{figure}

Let $\thresh(\idxuser) = (\threshord{1}(\idxuser), \ldots, \threshord{\idxuser-1}(\idxuser) )$ 
be the sequence of thresholds used by the algorithm after step $\idxuser-1$, and $\score$ the sequence of item scores. 

Note that $D_\idxbin(\idxuser) \sim \text{Beta}(1, \idxuser-1)$, because it is the difference between two consecutive order statistics of a $\mathcal{U}([0,1])$, cf Lemma \ref{lem:diffunif} (proven below in the Appendix). 

Let $k^*(\idxuser)$ the index of the bin returned by $\bisolate_{\idxuser}$. 
A bin $\bin_\idxbin$ can be selected by $\bisolate_{\idxuser}$ only if the threshold of user $\idxuser$ belongs to $[\score^-_\idxbin(\idxuser), \score^+_\idxbin(\idxuser)]$.
So, for all $\idxbin, \idxuser$, we have the inclusion:  

\begin{equation}
\label{eq:incl}
(k^*(\idxuser) = k) \subseteq (\thresh_u \in [\score^-_k, \score_k^+])
\end{equation}

We decompose the expected cost in queries of the $\bsplit$ phase ($\tsplit$) by the expected cost for each step $\idxuser$ of the algorithm ($\tsplit_{\idxuser}$). 
For each step, we further decompose $\tsplit_{\idxuser}$ by the expected cost for each bin of this step, indexed by $\idxbin$. \\
Note that we consider all pair of consecutive thresholds as a bin, even if the bin contains no item. 
Recall that $\bin_\idxbin(\idxuser)$ is the \emph{size} of the bin (\emph{i.e.} the number of item scores between $\threshord{\idxbin}(\idxuser)$ and $\threshord{\idxbin+1}(\idxuser)$). \\
Then we have:

\begin{align}
    \E[\tsplit_{\idxuser} \mid \thresh(\idxuser)]
    & = \sum_{\idxbin=0}^{\idxuser} \prob(k^*(\idxuser)=k \mid \thresh(\idxuser)) \E[ B_\idxbin(\idxuser) \mid k^*(\idxuser)=k, \thresh(\idxuser)] \\
    \label{lin:condition_include}
    & \leq \sum_{\idxbin=0}^{\idxuser} \prob( \thresh_u \in [\score^-_k, \score_k^+] \mid \thresh(\idxuser)) \E[ B_\idxbin(\idxuser) \mid \thresh_u \in [\score^-_k, \score_k^+], \thresh(\idxuser)] \\
    & = \sum_{\idxbin=0}^{\idxuser} \E[D^+_\idxbin(\idxuser) \mid \thresh(\idxuser)] \E[ B_\idxbin(\idxuser)  \mid  \thresh(\idxuser)] \\
    & = \sum_{\idxbin=0}^{\idxuser} \E[D^+_\idxbin(\idxuser) \mid \thresh(\idxuser)] \nbi D_\idxbin(\idxuser) \\
    & = \nbi \sum_{\idxbin=0}^{\idxuser} \E[D^+_\idxbin(\idxuser) \mid \thresh(\idxuser)]D_\idxbin(\idxuser) 
\end{align}

where \eqref{lin:condition_include} uses Lemma \ref{lem:includecond} (proven below in the Appendix) and equation \eqref{eq:incl}. \\

By definition, $D^+_\idxbin(u) = D_\idxbin + d^- + d^+$.
So 
\begin{align*}
\E[D^+_\idxbin(\idxuser) \mid \thresh(\idxuser)]D_\idxbin(\idxuser)
& =   \E[D_\idxbin(\idxuser) + d^+_k  + d^-_k  \mid \thresh(\idxuser)]D_\idxbin(\idxuser) \\
& =   D_\idxbin(\idxuser)^2 + \E[d^+_k +  d^-_k  \mid \thresh(\idxuser)]D_\idxbin(\idxuser) \\
\end{align*}

So, 

\begin{align*}
\E[\tsplit_{\idxuser} ]
&= \E[\E[\tsplit_{\idxuser} \mid \thresh(\idxuser)] ] \\
& \leq \E\left[ n \sum_{k=0}^u ( D_\idxbin(\idxuser)^2 + \E[d^+_k +  d^-_k  \mid \thresh(\idxuser)]D_\idxbin(\idxuser)) \right]\\
& = \E\left[ n (u+1) ( D_K(\idxuser)^2 + \E[d^+_K +  d^-_K  \mid \thresh(\idxuser)]D_K(\idxuser)) \right]\\
\end{align*}
where $K$ is uniform on $[u] \cup \{0\}$ and independent of the thresholds and items.

Conditional on set of thresholds $Y(u)$, regardless of their positions in $[0,1]$, the expected distance between any given threshold and its closest item score on the right (or the left) is roughly $\frac{1}{n+1}$.
The approximation is even more valid as the number of users $\nbu$ grows, because $K$ is the index of one of these thresholds selected uniformly at random, so $d_K^+$ is approximately the distance between two consecutive points selected uniformly and independently at random in $[0,1]$.

So we have $\E[d^+_K +  d^-_K  \mid \thresh(\idxuser)] 
\simeq \frac{2}{n+1}$.

\begin{align*}
\E[\tsplit_{\idxuser} ]
& \leq n (u+1) \left(\E[  D_K(\idxuser)^2 ]+ \E[ \E[d^+_K +  d^-_K  \mid \thresh(\idxuser)]D_K(\idxuser)) ] \right) \\
& \simeq n (u+1) \left(\E[D_K^2] + \frac{2}{n+1}\E[ D_K] \right) \\
&= n (u+1) \left(\frac{2}{(u+1)(u+2)} + \frac{2}{n+2} \frac{1}{u+1} \right) \\
& =  2\frac{n}{u+2} + 2\frac{n}{n+2}  \\
\end{align*}

We sum on all users to get the final result:

\begin{align*}
\E[\tsplit]
&= \sum_{u=1}^m \E[\tsplit_{\idxuser} ] \\
& \lesssim \sum_{u=1}^m \left(2\frac{n}{u+2} + 2\frac{n}{n+2} \right) \\
& \simeq 2 n \log(m) + 2 m
\end{align*}

\end{proof}

\section{Inequalities}
\label{app:sec:inequalities}

We regroup here a few general inequalities which are used in the proofs of other lemmas.

\begin{lemma}
\label{lem:includecond}
    Let $A \subseteq B$ be two probabilistic events and $X$ a positive random variable, then
    $$\prob(A) \E[X \mid A] \leq \prob(B) \E[X \mid B]$$
\end{lemma}

\begin{proof}
We have $A \subseteq B$, so $B = A \cup (B \backslash A)$. So:
\begin{align*}
\prob(B) \E[X|B] 
&= \prob(B) \E[X|A \cup (B \backslash A)]  \\
&= \prob(B) (\prob(A|B) \E[X|A ] + \prob(B\backslash A | B) \E[X | B \backslash A]) && \text{(because $A$ and $B\backslash A$ are disjoint)} \\
&\geq \prob(B) \prob(A|B) \E[X|A ]  && \text{(because  $X \geq 0$)}\\
& = \prob(A) \E[X|A] 
\end{align*}
\end{proof}

\begin{lemma}
\label{lem:binomial_bound}

For all $k \leq n \leq m$

$$
\frac{\binom{n}{k}}{\binom{k+m}{k}} 
\leq \frac{e^2}{2\pi} \left( \frac{n}{m} \right)^k 
$$
\end{lemma}

\begin{proof}

For the case $k = n \leq m$, we have

$$
\left( \frac{\binom{n}{k}}{\binom{k+m}{k}} \right)^{-1}
= \binom{n+m}{n} 
=  \prod_{j=1}^n \frac{m+j}{j} 
=  \prod_{j=1}^n \left( 1+  \frac{m}{j} \right)
\geq \prod_{j=1}^n \left( 1+ \frac{m}{n} \right)
\geq  \left( \frac{m}{n} \right)^n
\geq  \frac{2\pi}{e^2} \left( \frac{m}{n} \right)^k
$$

For remaining cases $k < n \leq m$,  we use the bounds 
$\sqrt{2\pi} n^{n+1/2}e^{-n}  \leq n! \leq e n^{n+1/2}e^{-n}$.

They give us

\begin{align*}
\frac{\binom{n}{k}}{\binom{k+m}{k}} 
& = \frac{n!}{(n-k)!} \frac{m!}{(k+m)!}  \\
& \leq \frac{e n^{n+1/2} e^{-n}}{\sqrt{2 \pi} (n-k)^{n-k+1/2} e^{-(n-k)}} \frac{e m^{m+1/2} e^{-m}}{\sqrt{2 \pi} (k+m)^{k+m+1/2} e^{-(k+m)}}  \\
& \leq \frac{e^{2}}{2\pi} \frac{ n^{n+1/2} }{ (n-k)^{n-k+1/2} } \frac{ m^{m+1/2} }{ (k+m)^{k+m+1/2} }  \\
& = \frac{e^2}{2\pi}  
\left(\frac{n-k}{n}\right)^{-n}  \left(\frac{m+k}{ m} \right)^{-m}
n^{1/2} (n-k)^{k-1/2} m^{1/2}(m+k)^{-k-1/2} \\
&\leq  \frac{e^2}{2\pi}
\left(1 - \frac{k}{n}\right)^{-n}  \left(1 + \frac{k}{ m} \right)^{-m}
\left( \frac{n-k}{m+k} \right)^k \\
& =  \frac{e^2}{2\pi} \left(\frac{n}{m}\right)^{k} \left(1 - \frac{k}{n}\right)^{k-n}  \left(1 + \frac{k}{ m} \right)^{-k-m}
\end{align*}

We just need to show that $\left(1 - \frac{k}{n}\right)^{k-n}  \left(1 + \frac{k}{ m} \right)^{-k-m} \leq 1$.
For this, we will show that its logarithm is negative.
We have

$$
\log \left(  \left(1 - \frac{k}{n}\right)^{k-n}  \left(1 + \frac{k}{ m} \right)^{-k-m} \right)
=   (k-n) \log\left(1 - \frac{k}{n}\right)  - (k+m) \log\left(1 + \frac{k}{m} \right)
$$

In addition, for $0<x \leq 1$, we have
$
- \log(1-x) \leq x + \frac{x^2}{2(1-x)}
$
and
$
\log(1+x) \geq x - \frac{x^2}{2}
$,
which respectively give us

$$
(k-n) \log(1-\frac{k}{n}) \leq  (n-k) \left(\frac{k}{n} + \frac{(\frac{k}{n})^2}{2(1-\frac{k}{n})} \right)
= k - \frac{k^2}{2n}
$$

$$
-(k+m)\log(1 + \frac{k}{m}) 
\leq -(m+k)  \left(\frac{k}{m} - \frac{(\frac{k}{m})^2}{2} \right)
= -k - \frac{k^2}{2m} + \frac{k^3}{2m^2}
$$

Summing the two terms, we obtain

$$
(k-n) \log(1-\frac{k}{n})  -(k+m)\log(1 + \frac{k}{m}) 
\leq  -\frac{k^2}{2n} - \frac{k^2}{2m} + \frac{k^3}{2m^2}
$$

For $n$ large enough, we have $k \leq n \leq m$, which gives $- \frac{k^2}{2n} + \frac{k^3}{2m^2} \leq 0$, so finally

$$
(k-n) \log\left(1 - \frac{k}{n}\right)  - (k+m) \log\left(1 + \frac{k}{m} \right)
\leq -\frac{k^2}{2m} 
< 0
$$

\end{proof}

\begin{lemma}
\label{lem:conditional_expectation}
Let $Z$ be a random variable and $\mathcal{A}$ be a probabilistic event.
 Then
$$ | \E[Z] - \E[Z|\mathcal{A}] | 
\leq (1 - \prob(\mathcal{A})) |\E[Z | \bar{\mathcal{A}}] - \E[Z | \mathcal{A}] |
$$
in particular, if $0 \leq Z \leq c_Z$, then

$$ | \E[Z] - \E[Z|\mathcal{A}] | 
\leq  (1 - \prob(\mathcal{A})) c_Z
$$
\end{lemma}

\begin{proof}

\begin{align*}
\E[Z] 
& = \prob(\mathcal{A}) \E[Z |\mathcal{A}] +  (1- \prob(\mathcal{A})) \E[Z |\bar{\mathcal{A}}] 
\end{align*}

so
\begin{align*}
| \E[Z] - \E[Z | \mathcal{A}] | 
& = |(\prob( \mathcal{A} ) - 1)   \E[Z |\mathcal{A}] + (1 -\prob( \mathcal{A}))    \E[Z |\bar{\mathcal{A}}] | \\
& =  (1- \prob(\mathcal{A})) | \E[Z |\mathcal{A}] -   \E[Z |\bar{\mathcal{A}}] |
\end{align*}

If $0 \leq Z \leq c_Z$, then
$
| \E[Z] - \E[Z | \mathcal{A}] | 
\leq c_Z
$,
so
\begin{align*}
| \E[Z] - \E[Z | \mathcal{A} ]| 
& =  (1- \prob(\mathcal{A})) c_Z
\end{align*}
\end{proof}

\end{document}